\documentclass{amsart}
\usepackage[margin=1.25in]{geometry}

\usepackage{amssymb,amsmath,graphicx}
\usepackage{amsthm}
\usepackage[mathscr]{euscript}
\usepackage{dsfont} 

\usepackage{thmtools} 
\usepackage{subfig} 
\usepackage[usenames,dvipsnames]{color}
\usepackage{hyperref}
\usepackage{cleveref} 

\def\bbX{\mathbb X}
\def\bbY{\mathbb Y}

\def\cH{\mathcal H}

\def\cU{\mathcal U}
\def\cX{\mathcal X}
\def\cY{\mathcal Y}
\def\cZ{\mathcal Z}
\def\zcf{\mathpzc f}
\def\zcg{\mathpzc g}
\def\bfe{\mathbf e}

\def\mfg{\mathfrak g}
\def\mfh{\mathfrak h}
\def\sfE{\mathsf E}
\def\sfP{\mathsf P}

\def\scE{\mathscr E}

\def\scI{\mathscr I}
\def\scV{\mathscr V}
\def\conv{\textup{conv}}
\newcommand{\N}{\mathbb{N}}
\newcommand{\R}{\mathbb{R}}
\newcommand{\ip}[2]{\left\langle #1, #2 \right\rangle}

\DeclareMathAlphabet{\mathpzc}{OT1}{pzc}{m}{it}

\newcommand{\Argmax}{\operatornamewithlimits{Argmax}}

\usepackage[usenames,dvipsnames]{color}

\title{Blackwell Approachability and Minimax Theory}
\author{Matus Telgarsky}
\email{mtelgars@cs.ucsd.edu}


\numberwithin{equation}{section} 
\declaretheorem[numberlike=equation]{theorem}
\declaretheorem[numberlike=theorem]{lemma}
\declaretheorem[numberlike=theorem]{proposition}

\declaretheorem[numberlike=theorem,style=definition]{definition}
\declaretheorem[numberlike=theorem,style=remark]{remark}

\newif\ifarxiv
\arxivtrue
\begin{document}
\maketitle

\begin{abstract}
This manuscript investigates the relationship between Blackwell Approachability,
a stochastic vector-valued repeated game, and minimax theory, a single-play scalar-valued
scenario.  First, it is established in a general setting --- one not permitting
invocation of minimax theory --- that Blackwell's Approachability 
Theorem (Blackwell~\cite{blackwell_approach}) and its
generalization due to Hou~\cite{hou_approach} 
are still valid.  Second,
minimax structure grants a result in the
spirit of Blackwell's 
weak-approachability conjecture (Blackwell~\cite{blackwell_approach}),
later resolved by Vieille~\cite{vieille_weak_approach},
that any set is either approachable by one player,
or avoidable by the opponent.  
This analysis also reveals a strategy for the opponent.
\end{abstract}

\section{Introduction.}
Consider a repeated game between two players, one selecting
$x_t\in \cX$, the other $y_t\in\cY$,
where the payoff is determined by
a vector-valued function $\zcf : \cX\times \cY \to \R^d$.
In Blackwell's Approachability
Game (Blackwell~\cite{blackwell_approach}), the player choosing from $\cX$ tries to keep
the center of gravity $t^{-1} \sum_{i=1}^t \zcf(x_i,y_i)$ arbitrarily close to 
some target set $S\subseteq\R^d$, whereas the opponent aims to deny this.

This game has many interesting applications:
the existence of Hannan consistent forecasters 
(Blackwell~\cite{blackwell_approach_2})
and calibrated forecasters (Foster and Vohra~\cite{foster_vohra}), to name a few.  
But on the abstract side, the basic quandary is: what is the equilibrium structure?

This was a central question of Blackwell's original manuscript, ``An analog of
the minimax theorem for vector payoffs'' (Blackwell~\cite{blackwell_approach}).
Both there, and in subsequent presentations, one may find invocations of standard minimax
theory; for instance, 
when 
establishing Blackwell's 
Approachability Theorem, which geometrically characterizes 
sets where the $\cX$-player can guarantee victory.
But the usual setting had
linear payoffs and convex compact domains, which are 
the sufficient conditions for 
von Neumann's Minimax Theorem.  The question thus remains: what is the real relationship 
between 
equilibria in Blackwell Games, and scalar-valued minimax theory?


The goal of this manuscript is to determine this dependence by working in a more 
general setting (i.e., a choice of $(\cX,\cY,\zcf)$ where minimax theorems may
simply fail).  The organization will be provided shortly, but the main results can be
summarized intuitively.  Recall that standard minimax theorems
--- usually providing $\min_{x\in\cX} \max_{y\in\cY} f(x,y) = 
\max_{y\in\cY}\min_{x\in\cX} f(x,y)$ for some family of sets and scalar-valued functions 
--- can be interpreted as stating that the order of the players does not affect
the dynamics.  This will be carried over to the vector-valued case.
\begin{itemize}
\item Hou's generalization of the Blackwell Approachability Theorem
(cf. Hou~\cite{hou_approach}, later rediscovered by Spinat~\cite{spinat_approach})
characterizes the approachable sets when the $\cX$-player
moves first;
it considers no other order, and thus holds without any appeal to minimax theory.
(Cf. \Cref{fact:blackwell_approach}.)
\item When minimax holds for certain scalar-valued subproblems, then the player order
may be reversed without changing the dynamics. This will be used to prove a result 
similar to 
Vieille's weak-approachability/weak-excludability theorem
(Vieille~\cite{vieille_weak_approach}): 
\Cref{fact:approach_avoid} will effectively state that minimax structure removes
the impact of order in approachability games.
\end{itemize}

The above characterization made glib reference to scalar-valued subproblems.  The most
natural appearance of these is to consider  halfspaces as target sets,
and specifically the scalar-valued
game which arises by projecting $\zcf$ onto the halfspace's normal.  Such scalarizations
have always appeared centrally in the discussion of Blackwell approachability --- indeed,
Blackwell's approach strategy (which was also used by Hou~\cite{hou_approach}),
is a greedy 
algorithm that chooses a halfspace orthogonal to the current projection onto the target
set,
and attempts to force a payoff inside it.  A side goal of this manuscript is to 
provide a deeper understanding of the role played by halfspaces;
this turns out to be a convenient detour, as these scalarizations grant an easy 
mechanism to trace the impact of (scalar-valued) minimax theory.

This manuscript is organized as a progression from minimax problems to the full
approachability game.  
Specifically, after presenting background in
\Cref{sec:setup}: deterministic single-play vector-valued games are in \Cref{sec:forcing};
deterministic repeated-play vector-valued games are in \Cref{sec:approaching}; finally,
stochastic repeated-play vector-valued games are in \Cref{sec:random}.

The single-play games of \Cref{sec:forcing},
though trivial, carry the essential ingredients of the
eventual message. First, minimax properties are only needed in the vector-valued
case where they are in needed in the scalar-valued cases: precisely when the
order of the players must be reversed.  Second,  this reversal only works for halfspaces:
even in the case of compact convex sets, the dynamics become unintelligible.

The repeated games of 
\Cref{sec:approaching} provide the heart of the matter.
First, this section develops a geometric characterization of
approachable sets in the spirit of Blackwell, Hou, and Spinat 
\cite{blackwell_approach,hou_approach,spinat_approach}, but in general spaces which 
may disallow the application of minimax theorems.
Second, in the spirit of weak-approachability/weak-excludability
results (Vieille~\cite{vieille_weak_approach}),
it provides a characterization of general sets as either approachable by one player,
or avoidable by the other.  This second result depends on minimax structure, and can
fail without it.  This section also presents a strategy for the opponent.

To close, \Cref{sec:random} confirms that studying the deterministic cases 
suffices: the stochastic and deterministic
settings have the same approachable sets.

\section{Background.}
\label{sec:setup}
Denote by $\ip{\cdot}{\cdot}$ and $\|\cdot\|$ the standard
inner product and norm on $\R^d$.  For two points $\phi,\psi\in\R^d$, 
overload interval notation for higher dimensions:
e.g., $[\phi,\psi] := \{\alpha\phi + (1-\alpha) \psi : \alpha \in [0,1]\}$.
Write
$\rho(\phi,\psi):= \|\phi-\psi\|$, and
for a set $S\subseteq \R^d$, define $\rho(\phi,S) := \inf_{\psi\in S} \rho(\phi,\psi)$.
Let $B(\phi,\epsilon)$ denote the closed ball of radius $\epsilon$,
and $S_\epsilon := S + B(0,\epsilon)$ be the $\epsilon$-neighborhood of $S$.
Recall that, for nonempty compact sets $S,U\subseteq \R^d$, the Hausdorff metric
$\Delta(S,U) := \inf \{\epsilon > 0 : S \subseteq U_\epsilon \land U\subseteq S_\epsilon\}$
is complete (cf. Munkres~\cite[Exercise 45.7]{munkres_topology}).
Finally, for any set $S$, let $S^o$, $\overline S$, and $S^c$ 
respectively denote the interior, closure, and complement of $S$.

This manuscript will consider both single-play vector-valued games,
termed {\em forcing} games, and repeated games, termed approachability games.
These games are characterized by a 4-tuple 
$(\cX, \cY, \zcf, S)$.
\begin{itemize}
\item
One player, the $\cX$-player or simply $\cX$, chooses $x\in\cX$;
the opponent (similarly $\cY$-player or just $\cY$) chooses $y\in\cY$.
\item
A bounded function $\zcf : \cX\times\cY \to \R^d$
with bound $\|\zcf(x,y)\| < \gamma$ determines the payoff.
\item
The $\cX$-player desires payoffs near a
target set $S\subseteq\R^d$; the $\cY$-player tries to avoid $S$.
\end{itemize}
Boundedness of $\zcf$ is necessary to the analysis; note that the strict inequality
prevents $\gamma = 0$, which is a trivial scenario anyway.

\begin{definition}
Say $g: \cX \times \cY \to \R$ has the {\em minimax property} when
there exist $\bar x \in \cX$, $\bar y \in \cY$ satisfying
\[
\sup_{y\in\cY} g(\bar x, y) \leq g(\bar x, \bar y) \leq \inf_{x\in\cX} g(x,\bar y).
\]
(This provides $\inf_x \sup_y g(x,y) = \sup_y \inf_x g(x,y)$.)
A function $f: \cX\times\cY \to \R^d$ has the {\em minimax property} when, 
for every $\lambda \in \R^d$, the function
$\ip{f(\cdot,\cdot)}{\lambda}$ has the minimax property.
\end{definition}
This latter property is quite restrictive, and \Cref{sec:minimax_families}
discusses function classes which satisfy it.  Reliance on the minimax property
will always be stated explicitly.  

The single-play game is defined as follows; note that, with the exception of
\Cref{sec:random}, the games in this manuscript are deterministic.
\begin{definition}
A {\em forcing game} has only one round, where
$\cX$ wins iff $\zcf(x,y) \in S$ ($\cY$ wins iff $\zcf(x,y)\in S^c$).
Say \emph{$\cX$ can force $S$ as player 1}, or more succinctly 
\emph{$\cX$ can 1-force $S$},
when $\exists x\in\cX\centerdot \forall y\in\cY\centerdot \zcf(x,y)\in S$.
The weaker property that 
$\cX$ can {\em 2-force} $S$ means
$\forall y\in \cY\centerdot\exists x\in\cX\centerdot \zcf(x,y)\in S$.
Analogously, one can define 1- and 2-forcing for $\cY$, where the goal is now $S^c$.
\end{definition}
This terminology reflects an ordering of player moves induced by the quantifiers.
Suppose $\cX$ can 1-force $S$; then even if $y\in\cY$ is selected with
knowledge of the chosen $\bar x\in\cX$, an $\bar x$ exists so that $\zcf(\bar x,y)\in S$.
On the other hand, if $\cY$ can 2-force $S^c$, and if this choice is with knowledge
of the selected $x \in \cX$, then the outcome $\zcf(x,y)\not\in S$ can be forced.
If $\cX$ can 1-force $S$, it can
win with either order.  But if $\cX$ can only 2-force $S$, then in general it can only
win as the player who moves last.

In the repeated game setting, players choose $(x_t)_{t=1}^\infty$ and $(y_t)_{t=1}^\infty$,
and have access to the selection history $\cH_t := ((x_i,y_i))_{i=1}^{t}$.
A strategy for $\cX$
is a family of functions $\mfg := (\mfg_t : (\cX\times\cY)^{t-1} \to \cX)_{t=1}^\infty$;
analogously, a strategy for $\cY$
is a family $\mfh := (\mfh_t : (\cX\times\cY)^{t-1} \to \cY)_{t=1}^\infty$.
In round $t+1$, the players are presented the current history 
$\cH_t := ((x_i,y_i))_{i=1}^t$, then select
$x_{t+1} := \mfg_{t+1}(\cH_t)$ and $y_{t+1} := \mfh_{t+1}(\cH_t)$.

Given a history $\cH_t$, define $\phi_t := t^{-1} \sum_{i=1}^t \zcf(x_i,y_i)$.
the goal in the (deterministic) 
repeated setting will be to manipulate this center of gravity,
where the averaging will suffice to make the family of approachable sets much richer
than the forcible sets.


\begin{definition}
A set $S$ (or, for clarity, a tuple $(\cX,\cY,\zcf,S)$) is \emph{approachable} when
\[
\exists \mfg
\centerdot
\forall \epsilon > 0
\centerdot
\exists T
\centerdot
\forall \mfh
\centerdot
\forall t \geq T
\centerdot
\phi_t \in S_\epsilon.
\]
Any $\mfg$ satisfying this definition is an {\em approach strategy}.
\end{definition}

In other words, $\lim_{t\to\infty} \inf_{z\in S} \|\phi_t - z\| = 0$, and this
convergence is uniform with respect to the family of opponent strategies.

Note that for any $\cX$-player strategy $\mfg$, there exists a family of opponent 
strategies which assume that $\cX$ is playing $\mfg$; these strategies effectively
choose $y_t$ knowing $x_t$.  Moreover, these strategies are considered in the universal 
quantification over opponent strategies in the definition of approachability.  
As such, when constructing an approach strategy, what the $\cX$-player can force in each
round are the 1-forcible sets; analogously, the opponent strategy is working with
2-forcible sets.
Said another way, the quantification order in the definition of approachability
implies another setting where $\cX$ moves first, and $\cY$ observes this before moving.

It is thus natural to consider the effect of reordering these quantifiers.  First,
one can say $S$ is not approachable when
\[
\forall \mfg
\centerdot
\exists \epsilon > 0
\centerdot
\forall T
\centerdot
\exists \mfh
\centerdot
\exists t \geq T
\centerdot
\phi_t \not \in S_\epsilon.
\]
Reversing the quantifiers for $\mfg$ and $\mfh$ yields a setting where $\cY$ moves first.
\begin{definition}
A set $S$ (or, for clarity, a tuple $(\cX,\cY,\zcf,S)$) is {\em avoidable} when
\[
\exists \mfh
\centerdot
\exists \epsilon > 0
\centerdot
\forall T
\centerdot
\forall \mfg
\centerdot
\exists t \geq T
\centerdot
\phi_t \not \in S_\epsilon.
\]
Any $\mfh$ satisfying this definition is an {\em avoidance strategy}.
\end{definition}

Mirroring the discussion of player order and approachability, an avoidance
game effectively has the $\cY$ player move first, and $\cX$ chooses with knowledge
of this move.  Correspondingly, in any given round, $\cY$ can force 1-forcible sets,
and $\cX$ has the easier criterion of 2-forcible sets.

Blackwell considered a stronger property than avoidability, called
{\em excludability}, where the quantifiers
$\exists \epsilon > 0$,
$\forall T$,
and 
$\exists t \geq T$ were 
respectively replaced with $\forall \epsilon > 0$,
$\exists T$, and $\forall t \geq T$, thus matching the goal of the
$\cX$-player.
While there exist games which are neither approachable
nor excludable (Blackwell~\cite{blackwell_approach}),
a weaker definition grants that every game
is either weak-approachable or weak-excludable (Vieille~\cite{vieille_weak_approach}).
\Cref{sec:approaching} will show that, under the minimax property, every game is either
approachable or avoidable;  as in the scalar-valued setting, minimax structure
nullifies the effect of player order.

A few final technical conveniences are in order.
Since $S$ is approachable iff its closure is approachable, this manuscript will
follow the usual convention of considering only the case that $S$ is closed.
Next, every superset $S'$ of an approachable set $S\subseteq S'$ is itself 
approachable, simply by running the approach strategy for $S$.  Combining this
with the fact that $\zcf$ is bounded, it suffices to consider only compact 
sets.  Lastly, define $\conv(\zcf(\cX,\cY))$ to be the convex hull of the range of
$\zcf$; notice that $\phi_t \in \conv(\zcf(\cX,\cY))$.

\section{Forcing Games.}
\label{sec:forcing}
A few properties are straight from the definitions.
\begin{itemize}
\item $\cX$ can 1-force $S$ iff $\cY$ can not 2-force $S^c$.
\item If $\cX$ can 1-force $S$, then it can 1-force $S'\supseteq S$.
\item $\cX$ can 1-force $S$ iff $S$ intersects every set $S'$ which can be
2-forced by $\cY$.
\item If $\cX$ can 1-force $S$, then $\cX$ can 2-force $S$.
\end{itemize}

Attempting to reverse this final property is where things become interesting.



\begin{proposition}
\label{fact:force_minimax}
Let any halfspace $H := \{z\in\R^d : \ip{\lambda}{z} \leq c\}$ be given,
and suppose $\ip{\zcf(\cdot,\cdot)}{\lambda}$ has the minimax property.
If $\cX$ can 2-force $H$, then $\cX$ can 1-force $H$.
\end{proposition}

\begin{proof}
Given any $y\in\cY$, choose $x_y\in\cX$ satisfying $\zcf(x_y,y) \in H$.
Since for all $y$
\[
c\geq \ip{\zcf(x_y,y)}{\lambda} \geq \inf_{x\in\cX} \ip{\zcf(x,y)}{\lambda},
\]
it follows that 
\[
c\geq \sup_{y\in\cY} \inf_{x\in\cX} \ip{\zcf(x,y)}{\lambda} = 
\sup_{y\in\cY}  \ip{\zcf(\bar x,y)}{\lambda},
\]
where the existence of $\bar x$ is from the minimax property.
It follows that for every $y\in\cY$,
$\zcf(\bar x, y)\in H$.
\end{proof}

To see that minimax structure is necessary, it suffices to consider a game
in $\R^1$ with $a := \inf_x \sup_y \zcf(x,y) > \sup_y\inf_x \zcf(x,y) =: b$; in particular,
the set $H := (-\infty, (a+b) / 2]$ can be 2-forced by $\cX$, but not 1-forced.
This problem can then be lifted to any dimension $d>1$ by constructing $\zcf':\cX\times\cY\to\R^d$
with $\ip{\zcf'(\cdot,\cdot)}{\lambda}= \zcf(\cdot,\cdot)$, and considering the set
$H' := \{z\in\R^d : \ip{\lambda}{z} \leq (a+b)/2\}$.
This final example also demonstrates how the restriction to halfspaces fits in: the game is effectively
projected onto the halfspace normal, thus becoming scalar-valued.

So the natural question is: what can be said about sets which are not halfspaces?

\subsection{Vector-valued Games as 3-player Games.}
Even in the case of minimax structure and compact convex target sets,
there are 2-forcible sets which can not be
1-forced: such an example appears in \Cref{sec:example}.  The goal of this subsection
is to investigate why these difficulties arise.

First note the following characterization of (closed convex) set membership.
\begin{proposition}
\label{fact:conv_force}
Let a closed convex nonempty set $S$ and any point $\phi\in \R^d$ be given.
Let 
$\sigma_S(\lambda) = \sup_{z\in S} \ip{z}{\lambda}$ denote the support function of $S$.
Then
\begin{align*}
\phi \in S
\qquad\iff\qquad
\sup_{\substack{\lambda \in \R^d\\\|\lambda\|\leq 1}} \ip{\phi}{\lambda}
-\sigma_S(\lambda) = 0.
\end{align*}
\end{proposition}
As $\phi\in S$ is equivalent to $\inf_{z\in S} \|\phi - z\|^2 = 0$,
this statement can be understood via convex duality.
\begin{proof}
In general, $\sup_{\|\lambda\|\leq 1} \ip{\phi}{\lambda} - \sigma_S(\lambda) \geq \ip{\phi}{0} - \sigma_S(0) = 0$.
Suppose $\phi\in S$; then for any $\lambda$, 
$\ip{\phi}{\lambda} \leq \sup_{\phi\in S} \ip{\phi}{\lambda} = \sigma_S(\lambda)$,
thus $\sup_\lambda \ip{\phi}{\lambda} - \sigma_S(\lambda) \leq 0$.  On the
other hand, if $\phi\not\in S$, there exists a halfspace
$H= \{z\in \R^d:\ip{\lambda'}{z} \leq c\}$ satisfying $S\subseteq H$ and
$\phi \not \in H$, meaning
$\sigma_S(\lambda') \leq  \sigma_H(\lambda') = c < \ip{\phi}{\lambda'}$
so $\ip{\phi}{\lambda'}  - \sigma_S(\lambda') > 0$.
\end{proof}

For any closed convex set $S$,
consider the scalar valued function
\[
\zcg_S(x,y,\lambda) := \ip{\zcf(x,y)}{\lambda} - \sigma_S(\lambda).
\]
Set $\cZ := \{\lambda\in\R^d : \|\lambda\|\leq 1\}$;
\Cref{fact:conv_force} grants
\begin{align*}
\textup{$\cX$ can 1-force $S$}
&\quad\Longrightarrow\quad
\inf_{x\in\cX} \sup_{y\in\cY} \sup_{\lambda\in\cZ} \zcg_S(x,y,\lambda) = 0,
\\
\textup{$\cX$ can 2-force $S$}
&\quad\Longrightarrow\quad
\sup_{y\in\cY} \inf_{x\in\cX} \sup_{\lambda\in\cZ} \zcg_S(x,y,\lambda) = 0.
\end{align*}
In general, $\inf_x \sup_y \sup_\lambda \zcf(x,y,\lambda) \geq 
\sup_y \inf_x \sup_\lambda \zcf(x,y,\lambda)$.
But when this does not hold with equality for a 2-forcible $S$,
one has an example which is 2-forcible but not 1-forcible.  The remaining discussion
thus considers this expression.

Suppose now that $\cX,\cY$ are convex and compact, and that $\zcf(\cdot,y)$ is convex
for every $y\in \cY$, and $\zcf(x,\cdot)$ is concave for every $x\in\cX$.
Since $\ip{\cdot}{\cdot}$ is 
bilinear and $\sigma_S$ is a convex function
(cf. Hiriart-Urruty and Lemar\'echal~\cite[Proposition B.2.1.2]{HULL}),
$\zcg_S$ is concave in $\lambda$.

Indeed, $\zcg_S$ can be viewed as providing the payoff for a 3-player game.  The
aforementioned structure suffices to grant one type of equilibrium: in particular,
the existence of Nash Equilibria
(cf. Borwein and Lewis~\cite[Exercise 8.3.10.d]{borwein_lewis}
for a suitable generalization).
But that only grants that players have no incentive to deviate without collusion,
whereas here  $y$ and $\lambda$ are effectively cooperating.
The function $\sup_{\lambda}\zcg_S(x,\cdot,\lambda)$ is highly nonconvex,
and thus lacks the usual structure allowing a statement
like $\inf_x \sup_y \sup_\lambda \zcg_S(x,y,\lambda) = \sup_y \inf_x \sup_\lambda \zcg_S(x,y,\lambda)$.
Thus, in general, there is a gap between the two sides of this expression.

This third player choosing $\lambda$ thus introduces major difficulty into the
problem.  Note that in the repeated game, only one $\lambda$ is chosen (not
a sequence, as with $x_t$ and $y_t$).  
This suffices to make the approachable sets far different from the forcible
sets.  Furthermore, the strategies for both players can be seen as attempting
to work with or against this third player $\lambda$, whose maximizing choice is
the direction of projection onto $S$, equivalently a hyperplane tangent to $S$.

\section{Approachability Games.}
\label{sec:approaching}


Halfspaces continue to be central in the repeated game setting; the following 
definition provides the basic tool whereby a halfspace is useful to either player.


\begin{definition}
\label{defn:hforce}
Consider, as in \Cref{fig:hforce}, a pair of points $\psi\in S$ and 
$\phi\in \conv(\zcf(\cX,\cY)) \setminus S$
with $\rho(\phi,\psi) = \rho(\phi,S)$, 
and a halfspace $H$ orthogonal to and also passing through $[\phi,\psi)$.
Refer to any tuple $(\phi,\psi,H)$ 
satisfying this arrangement as a
\emph{halfspace-forcing candidate for $S$}.
If $\cX$ can 1-force $H$, call this a \emph{halfspace-forcing
example for $S$}.  When $\cX$ can not 1-force $H$ (meaning $\cY$ can 2-force $H^c$),
call this tuple a \emph{halfspace-forcing counterexample for $S$}.
\end{definition}

\begin{figure}[]
\centering
\includegraphics[width=0.4\textwidth]{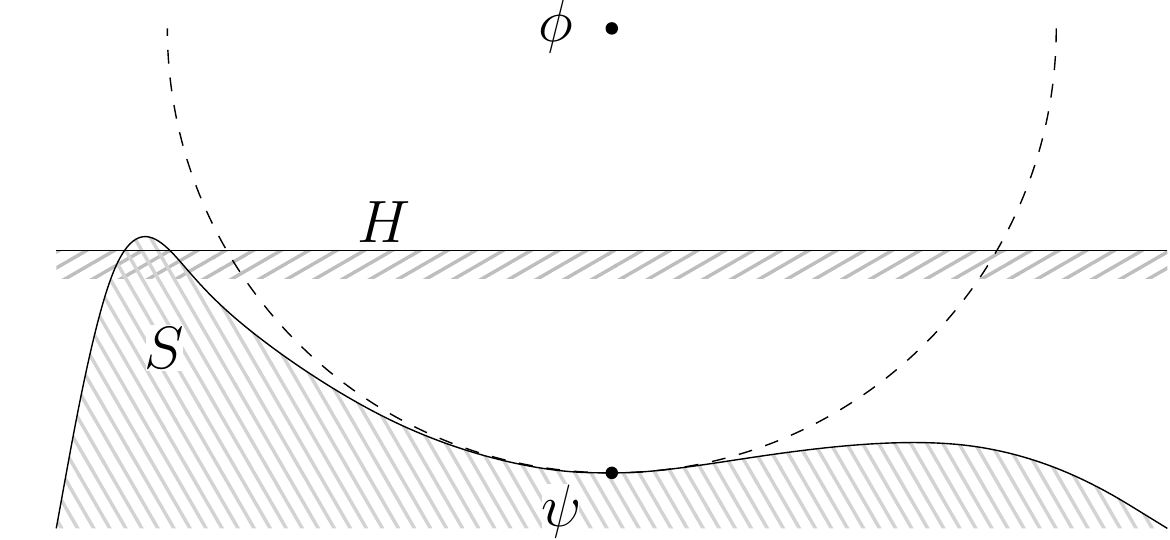}
\caption{Depiction of \Cref{defn:hforce}.}
\label{fig:hforce}
\end{figure}

Consider a halfspace-forcing candidate $(\phi,\psi,H)$ for some compact set $S$.
If this is a halfspace-forcing example,
and $\phi = \phi_t = t^{-1} \sum_{i=1}^t \zcf(x_i,y_i)$ for some iteration $t$,
then the $\cX$ player may force $H$ and move closer to $S$; this will be proved
in \Cref{fact:force}, and will provide the basis for the approach strategy $\mfg^*$.
On the other hand, if this is a halfspace-forcing counterexample, and
$\phi_t$ is close to $\psi$ for some iteration $t$, then the opponent may force $H^c$
and move away from $S$; this is proved in \Cref{fact:antiforce}, and provides the
basis for the opponent strategy $\mfh^*$.

Halfspace-forcing examples were used by Blackwell~\cite{blackwell_approach} 
to construct the original greedy approach strategy (and were again used by
Hou~\cite{hou_approach} and Spinat~\cite{spinat_approach});
the only distinction here is that
the halfspaces need not touch the target $S$.
Where the present work truly departs from earlier works is by also developing
a theory of unapproachable sets with halfspaces as a starting point.  This will
allow the construction of an opponent strategy $\mfh^*$, and also make it easy
to measure the impact of minimax structure (since, combined with \Cref{fact:force_minimax}
and properties of forcible sets, if $(\phi,\psi,H)$ is a halfspace-forcing
counterexample, then $\cY$ may equivalently 1-force and 2-force $H^c$),
culminating in the approachability/avoidability statement of 
\Cref{fact:approach_avoid}.
The following definition, which extends halfspace-forcing into a global property, will
characterize the approachable sets in \Cref{fact:blackwell_approach}.

\begin{definition}
A set $S\subseteq \R^d$ is an \emph{A-set} if
it has no halfspace-forcing counterexamples.
\end{definition}

In the case that $\cX$ is compact and $\zcf(\cdot,y)$ is continuous for 
every $y\in\cY$, an A-set $S$ will also satisfy a stronger property that
halfspace-forcing candidates $(\phi,\psi,H)$ may place $H$ tangent to $S$. 
Indeed, this grants a type of set called a {\em B-set} by Spinat~\cite{spinat_approach},
which was used there to characterize approachable sets.
In that setting, however, the compactness and continuity were guaranteed; in the 
more general choice of $(\cX,\cY,\zcf)$ here, the relaxed notion of A-set 
is necessary and sufficient.

\subsection{Sufficient Conditions for Approachability.}
First, a quantification of the progress granted by a single halfspace-forcing
example.
\begin{lemma}
\label{fact:force}
Let $(\phi,\psi,H)$ be a halfspace-forcing example for $S$,
and set $\tau := \rho(\psi,H^c)$.
Then there exists $\bar x\in\cX$ so that, for any $y\in\cY$,
\[
\ip{\zcf(\bar x, y) - \psi}{\phi - \psi} \leq \tau \gamma.
\]
\end{lemma}
\begin{proof}
Set $\psi'$ to be the projection of $\phi$ onto $H$,
and choose any $\bar x\in\cX$ providing 1-forcibility of $H$.
Then, for any $y\in \cY$,
\[
\ip{\zcf(\bar x,y) - \psi}{\phi - \psi}
=
\ip{\zcf(\bar x,y) - \psi'}{\phi - \psi}
+ \ip{\psi' - \psi}{\phi - \psi}
\leq 0 + \tau\gamma,
\]
which made use of $\|\phi-\psi\|\leq \gamma$.
\end{proof}

This \namecref{fact:force} leads to the following greedy strategy
$\mfg^* = (\mfg^*_t)_{t=1}^\infty$ for the approach 
player, parameterized by a family of tolerances $(\tau_t)_{t=1}^\infty$
(these tolerances being the only modification to
the strategy
provided by Blackwell, Hou, and Spinat
\cite{blackwell_approach,hou_approach,spinat_approach}).
\begin{definition}
The $\cX$-player strategy $\mfg^* = (\mfg^*_t)_{t\geq 1}$ is
\[
\mfg_{t+1}^*(\cH_t)
:=
\begin{cases}
x &
\exists x,\psi_t,H_t \centerdot
(\phi_t,\psi_t,H_t)\textup{ is a halfspace-forcing example for $S$},\\
& \qquad\rho(\psi_t,H_t^c) \leq \tau_t,
\forall y\in\cY \centerdot \zcf(x,y)\in H;\\
\textup{arbitrary} & \textup{otherwise.}
\end{cases}
\]
\end{definition}

\begin{theorem}
\label{fact:approach_sufficient}
Let an A-set $S$,
any tolerances $(\tau_t)_1^\infty$ with $\sum_{i\geq 1} \tau_i \leq \gamma$
and $\tau_t > 0$,
and any $\epsilon > 0$ be given.
If $\cX$ uses strategy $\mfg^*$,
then for any opponent strategy and 
any $t \geq 3\gamma^2 / \epsilon^2$,
$\phi_t \in S_\epsilon$.
\end{theorem}

The proof reveals that $\sum_{i=1}^{t} \tau_i = o(t)$ suffices; requiring a constant 
bound is for convenience.

\begin{proof}
So that $\psi_t$ is always defined, set
$\psi_t := \phi_t$ when $\phi_t \in S$.
It will be shown by induction that 
\[
\|\phi_t - \psi_t\|_2^2 \leq \frac {\gamma^2  + 2\gamma \sum_{i=1}^{t-1}\tau_i}{t};
\]
the result follows by
applying the bounds for $\sum_i \tau_i$ and $t$, then taking a square root.

In the base case,
$
\|\phi_1 - \psi_1\| \leq \gamma
$.
For the inductive step,
\begin{align}
\|\phi_{t+1} - \psi_{t+1}\|_2^2
&\leq  \|\phi_{t+1} - \psi_t\|_2^2
\notag\\
&= \frac 1 {(t+1)^2}\left\|
t(\phi_t - \psi_t) + (\zcf(x_{t+1},y_{t+1}) - \psi_t)
\right\|_2^2
\notag\\
&=
\frac {1}{(t+1)^2}\left(
t^2\|\phi_t-\psi_t\|_2^2 + 2t \ip {\phi_t - \psi_t}{\zcf(x_{t+1},y_{t+1}) - \psi_t} 
+ \|\zcf(x_{t+1},y_{t+1}) - \psi_t\|_2^2
\right).
\label{eq:approach:necessary:nofix}
\end{align}
If $\phi_t\in S$, then $\phi_t-\psi_t = 0$ and the middle term vanishes.
Otherwise, 
since $\tau_t > 0$ and $S$ is an A-set, there must exist a halfspace-forcing
example $(\phi_t,\psi_t,H_t)$ with $\rho(\psi_t,H_t^c)\leq \tau_t$,
and \Cref{fact:force} grants
$\ip{\phi_t - \psi_t}{\zcf(x_{t+1},y_{t+1}) - \psi_t} \leq \tau_t\gamma$.
Applying the inductive hypothesis,
\[
\eqref{eq:approach:necessary:nofix}
\leq
\frac {t\gamma^2 + 2\gamma t\sum_{i=1}^t \tau_t + \gamma^2}{(t+1)^2}
\leq
\frac {\gamma^2 + 2\gamma \sum_{i=1}^t \tau_t}{t+1}.
\]
\end{proof}

\begin{remark}
\label{rem:tricky_superset_approach}
This statement also grants the approachability of any set $S'$ which 
contains an A-set $S$: simply run $\mfg^*$ on $S'$.  But it is 
unclear how to produce an approach strategy for $S'$ directly.  Suppose
$S'$ is the union of an A-set, and another set which is nearly an A-set: a small
piece is missing in such a way to render it unapproachable.  The approach strategy
must somehow rule out gravitating towards the second, damaged set; detecting this
difficulty does not appear tractable.  Please see the examples of
\Cref{sec:example} for further discussion.
\end{remark}

\subsection{Sufficient Conditions for Non-approachability.}

The first step is complementary to \Cref{fact:force}: quantifying how much
a single halfspace-forcing counterexample benefits the $\cY$-player.

\begin{lemma}
\label{fact:antiforce}
Let $(\phi,\psi,H)$ be a halfspace-forcing counterexample.
Set $\tau := \rho(\psi,H^c)>0$,
$\epsilon := \epsilon(\tau) :=
\frac {\tau^2(\sqrt{\gamma^2 + \tau^2} -\gamma)}{8(4\gamma^2 + \tau^2)}$
and let $T \geq \lceil 8/\epsilon\rceil$ be given.
Then for any $p \in B(\psi,2\epsilon)$, and any sequence $(x_i)_{i\geq 1}$,
there exist $(y_i)_{i\geq 1}$ and $M \leq \lceil T\gamma\epsilon/8\rceil$
where
\[
\rho\left(
\phi, 
(T+M)^{-1}
(Tp + \sum\nolimits_{i=1}^M \zcf(x_i,y_i))
\right)
\leq \rho(\phi,S) - \epsilon.
\]
Under the stronger condition that $H^c$ may be 1-forced by $\cY$, there exists 
a single $\bar y$ so that the choice $y_i = \bar y$ grants the above property
for any sequence $(x_i)_{i\geq 1}$.
\end{lemma}

\begin{figure}[]
\centering
\includegraphics[width=0.5\textwidth]{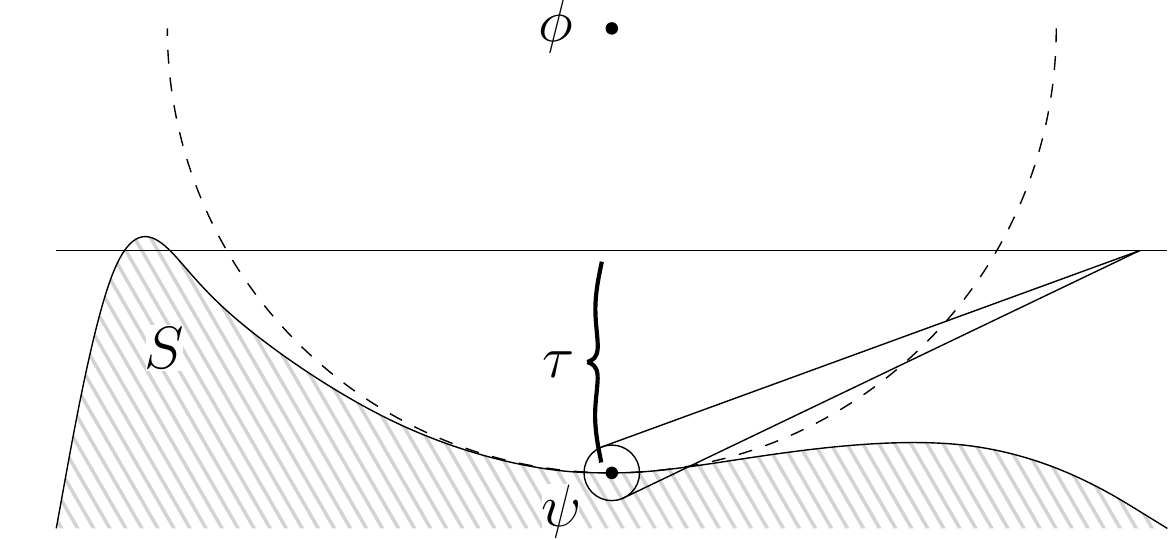}
\caption{Depiction of \Cref{fact:antiforce}.}
\label{fig:antiforce}
\end{figure}

This derivation is mechanical, and thus pushed to \Cref{sec:antiforce_proof},
but the idea, which appears in \Cref{fig:antiforce}, is simple.
By assumption, there exists a point $\psi\in S$
satisfying $\psi \in B(\phi,\rho(\phi,S)) \cap S$, and a halfspace $H$ whose boundary
is $\tau$ away from $\psi$, and $H$ is not 1-forcible by $\cX$.
Correspondingly, $H^c$ is 2-forcible by $\cY$.

To see how this helps $\cY$, by properties of $l^2$ balls and boundedness of
$\zcf$, every line connecting $\psi$ to
$H^c\cap \zcf(\cX,\cY)$ must pass interior to $B(\phi,\rho(\phi,S))$.
This remains true for a tiny neighborhood around $\psi$. Thus
it suffices for the $\cY$ player to force points in $H^c$: regardless of how the
$\cX$ player chooses, a future center of gravity will eventually land interior to $B(\phi,\rho(\phi,S))$.
The extra work in the lemma is in producing $\epsilon$ as a function of $\tau$,
allowing uniform controls for various halfspace-forcing counterexamples.
\begin{definition}
Define $\epsilon(\tau)$
as in the statement of \Cref{fact:antiforce}.
For any set $S\subseteq \R^d$,
define the {\em excess of $S$ with tolerance $\tau>0$} as
\[
\scE_\tau(S) :=
\Big\{B(\psi,\epsilon(\tau))^o :
\exists \phi,\psi,H \centerdot
(\phi,\psi,H)\textup{ is a halfspace-forcing counterexample with }
\rho(\psi,H^c) \geq \tau
\Big\}.
\]
Define $\scV_\tau(S) := S \setminus \bigcup_{U\in\scE_\tau(S)} U$.
\end{definition}

As per the following \namecref{fact:shrinkage_operator_approachability}, the operator
$\scV_\tau$ removes points which can not be used by an approach strategy.  A similar 
tool and subsequent limiting argument
appeared in the analysis of Hou~\cite{hou_approach} and Spinat~\cite{spinat_approach},
albeit without $\tau$, which will be used in constructing the opponent strategy.

\begin{lemma}
\label{fact:shrinkage_operator_approachability}
For any $\tau>0$ and $S\subseteq \R^d$,
$S$ is approachable iff $\scV_\tau(S)$ is approachable.
\end{lemma}

\begin{proof}
$(\Longrightarrow)$ Suppose $S$ is approachable, but contradictorily that 
$S' := \scV_\tau(S)$ is not approachable. Necessarily, $\scV_\tau(S) \neq S$ and
$\scE_\tau(S)$ is nonempty.  Let $\mfg$ be any approach strategy for
$S$; it must fail to approach $S'$, and therefore there must exist an opponent strategy
$\mfh$ such that $\phi_t \in S \setminus S'$ for
infinitely many $t$.
Now choose $T_0 \geq \lceil 8/\epsilon(\tau)\rceil$,
and let $T_1$ be the value provided by $\mfg$ guaranteeing 
$\phi_t\in S_{\epsilon(\tau)}$ for all opponent strategies and 
$t\geq T_1$.  Finally, set $T:=\max\{T_0,T_1\}$.
Consider the modified strategy $\mfh'$ which executes as $\mfh$ until
some $t> T$ satisfies $\phi_t\in S\setminus S'$.  Thereafter, it follows
the choices granted by \Cref{fact:antiforce} (with $p = \phi_t$),
and thus guarantees the existence
of $t' \geq T$ with $\phi_{t'}\not\in S_{\epsilon(\tau)}$,
contradicting the approachability
of $S$ by $\mfg$.  But $\mfg$ was arbitrary, and thus $S$ is not approachable.

$(\Longleftarrow)$ Supersets of approachable sets are always approachable.
\end{proof}
\begin{theorem}
\label{fact:shrinkage_operator_iteration}
Let compact $S\subseteq\R^d$ be given,
set $S_0 := S$ and $S_{i+1} := \scV_{1/(i+1)}(S_i)$.
Then the limit $S_\infty$ (in Hausdorff metric) exists.
Furthermore, exactly one of the following statements holds:
\begin{enumerate}
\item
There exists $N\in\N$ with $S_n=\emptyset$ for all $n\geq N$, and each $S_i$ is not approachable;
\item
$S_\infty$ is a compact nonempty A-set.
\end{enumerate}
\end{theorem}
The tolerance $1/(i+1)$ will make it easy to control the behavior of the eventual
opponent strategy $\mfh^*$.
\begin{proof}
Suppose there exists some $N$ such that $S_N=\emptyset$; since $\scV$ only makes sets
smaller, it follows that $S_n = \emptyset$ when $n\geq N$, providing $S_\infty = \emptyset$.
Next, the empty
set is never approachable, thus $S_n$ is not approachable when $n\geq N$.  
But \Cref{fact:shrinkage_operator_approachability} grants that $S_i$ is approachable
iff $S_{i+1}$ is approachable, and so by $N$ applications of this 
\namecref{fact:shrinkage_operator_approachability}, it follows that every $S_i$ is not
approachable.

Now suppose there does not exist any $N$ with $S_N=\emptyset$.
Thus each $S_i$ is compact and nonempty,
and $S_\infty$ exists and is a compact nonempty set by the completeness
of the Hausdorff metric on compact nonempty sets.


Assume contradictorily that $S_\infty$ is not an A-set,
meaning it has a halfspace-forcing counterexample $(\phi,\psi,H)$,
and set $\tau := \rho(\psi,H^c) >0$.
By \Cref{fact:nested_balls_redone}, there exists $\delta > 0$ so that
any set $S'$ satisfying $\Delta(S',S_\infty) < \delta$
has a halfspace-forcing counterexample $(\phi',\psi',H')$ with
$\rho(\psi',(H')^c) \geq \tau/4$.

Choose $j$ 
so that $\Delta(S_\infty, S_j) < \min\{\delta,
\epsilon(\tau/4)/2\}$.
Thus
$\scE_{1/(j+1)}(S_j)$ contains $B(\psi',\epsilon(\tau/4))^o$.
But $\scV$ only shrinks sets, meaning 
$\Delta(S_\infty,S_j) \geq \Delta(S_{j+1},S_j) \geq \epsilon(\tau/4)$, a contradiction.
\end{proof}

Consider the case that $S$ is not approachable; the above 
\namecref{fact:shrinkage_operator_iteration} grants the sequence $(S_i)_{i=0}^N$
with $S_0 = S$ and $S_N = \emptyset$.
Now define $E_i := S_i \setminus S_{i+1}$ for $i \in (0,\ldots,N-1)$,
and note that each $E_i$ is disjoint, and $S = \cup_{i=1}^{N-1}E_i$.
In this way, the operator $\scV$ peels $S$ into a finite sequence of concentric sets,
like the rinds of an onion; the strategy $\mfh^*$, depicted in 
\Cref{fig:mfh_star},
is to use \Cref{fact:antiforce} to move through these rinds, eventually
exiting $S$ entirely.  It is interesting to contrast the complexity of $\mfh^*$
with the triviality  of $\mfg^*$.

\begin{definition}
Fix any compact set $S$, and let $(S_i)_{i\geq 0}$ be as in 
\Cref{fact:shrinkage_operator_iteration}.
When there exists $N$ such that $S_N=\emptyset$, define
$\delta(S) := \epsilon(1/N)$ and 
$T(S) := \lceil 8 / \delta(S)\rceil$; otherwise, $\delta(S) := 0$ and
$T(S) := \infty$.  Finally, for any $\phi$, 
\[
\scI_S(\phi) := \begin{cases}
-1
& \textup{when $\delta(S) > 0$ and $\phi \not \in S_{\delta(S)}$,}
\\
\max \{ i : \phi \in (S_i)_{\delta(S)}\} 
& \textup{when $\delta(S) > 0$ and $\phi \in S_{\delta(S)}$,}
\\
\infty & \textup{otherwise.}
\end{cases}
\]
\end{definition}

\begin{definition}
The $\cY$-player strategy $\mfh^* = (\mfh^*_t)_{t\geq 1}$ is 
\[
\mfh_{t+1}^*(\cH_t) :=
\begin{cases}
\textup{arbitrary} 
&
\textup{when $t < T(S)$};
\\
y
&
\textup{when $t = T(S)$, or $t> T(S)$ and $\scI_S(\phi_t) \neq \scI_S(\phi_{t-1})$,}\\
&\qquad\textup{there exists a halfspace-forcing counterexample $(\phi,\psi,H)$}\\
&\qquad\quad\textup{with $\rho(\phi_t,\psi) \leq 2\epsilon(1/(\scI_S(\phi_t)+1))$
and
$\rho(\psi,H^c) \geq 1/(\scI_S(\phi_{t-1})+1)$,}\\
&\qquad\textup{choose any $y$ given by \Cref{fact:antiforce} (with $p=\phi_t$),}\\
&\qquad\quad\textup{ satisfying 1-forcing by $\cY$ if possible;}
\\
y'
&
\textup{when $t >  T(S)$ and $\scI_S(\phi_t) = \scI_S(\phi_{t-1})$,}\\
&\qquad\textup{choose $(\phi,\psi,H)$ and $p$ as for $\phi_{t-1}$,}\\
&\qquad\textup{choose $y'$ according to \Cref{fact:antiforce}, again trying to 1-force.}
\end{cases}
\]
\end{definition}

\begin{figure}[]
\centering
\includegraphics[width=0.4\textwidth]{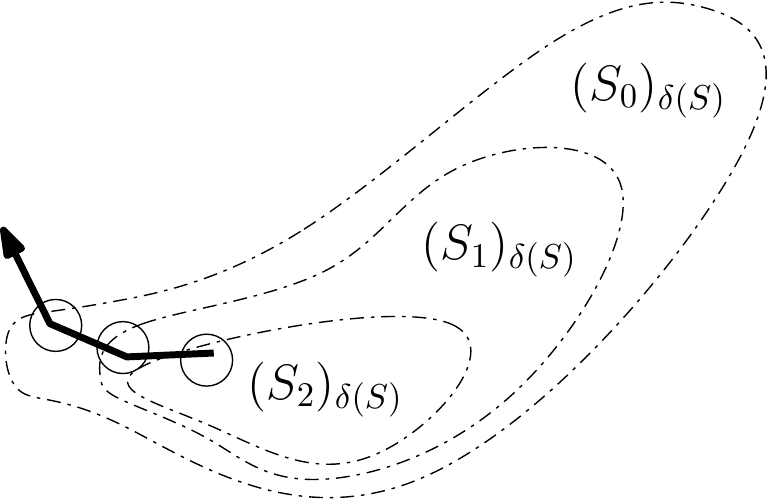}
\caption{$\cY$-player strategy $\mfh^*$, escaping $S$ rind by rind.  The small circles
denote the regions surrounding specific halfspace-forcing counterexamples, within 
which an iterate may fall and still make use of the counterexample (cf. \Cref{fact:antiforce}).
These balls are depicted as exceeding the onion rinds since \Cref{fact:antiforce}, which
determines $\delta(S)$, only
provides a lower bound on their size.  The arrow, abstractly indicating the
path out of $S$, may consist of a tremendous number of iterations.}
\label{fig:mfh_star}
\end{figure}

A similar mechanism was used by Hou~\cite[proof of Theorem 3]{hou_approach}
when proving necessary
and sufficient conditions for approachability (that is, the nonconvex form of
Blackwell's Approachability Theorem, presented here as \Cref{fact:blackwell_approach}).
There, halfspaces were not used as a primitive tool for the opponent: rather, 
the strategy was built up by considering a more abtract set guaranteeing progress
to the opponent (cf. Hou's \emph{insufficient subsets}~\cite[Definition 1]{hou_approach}).
Using halfspaces 
grants an arguably constructive opponent strategy, but more importantly
provides the backbone for measuring the effect of minimax structure.

\begin{proposition}
\label{fact:mfh_star}
Let $S\subseteq \R^d$ be given, and suppose $T(S)<\infty$.
Then for any approach strategy and any $T$, there exists $t\geq T$
so that, when $(y_i)_{i\geq 1}$ is chosen
by $\mfh^*$,
\[
\phi_t \not \in S_{\delta(S)}.
\]
Moreover, if every halfspace-forcing 
counterexample $(\phi,\psi,H)$ encountered by $\mfh^*$ is 1-forcible by $\cY$,
then $\mfh^*$ is an avoidance strategy.
\end{proposition}
\begin{proof}
From 
$T(S) < \infty$ it follows that $\delta(S) > 0$, and for any $\phi\in S_{\delta(S)}$,
$\scI_S(\phi) < \infty$. 
Suppose contradictorily that there exists $T \geq T(S)$
so that $\phi_t \in S_{\delta(S)}$ for all $t\geq T$.
In particular, this means 
$\min_{t\geq T} \scI_S(\phi_t) =: j > -1$.

Let $t\geq T(S)$ be the earliest iteration with $\scI_S(\phi_t) = j$,
which by definition of $\scI_S$ and the fact $S_i \supseteq S_{i+1}$ 
grants $\phi_t \in (S_j)_{\delta(S)} \setminus (S_{j+1})_{\delta(S)}$.
This $\phi_t$ satisfies the middle condition in the definition of $\mfh^*$;
it must be verified that the conditions for \Cref{fact:antiforce},
namely the existence of a halfspace-forcing counterexample and
nearby point $p$, are actually satisfied.

There are two cases for the location of $\phi_t$: either it is 
in $\scE_{1/(j+1)}(S_j)$, or it is in $(\scE_{1/(j+1)}(S_j))_{\epsilon(S)}$.  Recall that
$\scE_{1/(j+1)}(S_j)$ is the union of balls of radius at least
$\epsilon(1/(j+1)) \geq \epsilon(1/N) = \delta(S)$,
where the center $\psi$ of each can be made into a halfspace-forcing
counterexample $(\phi,\psi,H)$.  It thus suffices to take $p=\phi_t$:
when $\phi_t\in\scE_{1/(j+1)}(S_j)$, just take some ball in 
$\scE_{1/(j+1)}(S_j)$ containing 
$\phi_t$,
otherwise choose the closest ball, and the triangle inequality gives what is
needed 
(since the conditions on \Cref{fact:antiforce} allow distances up to
$2\epsilon(1/(j+1))$, and $\delta(S) \leq \epsilon(1/(j+1))$).

As such, $\scI_S(\phi_{t+1}) \leq \scI_S(\phi_t)$.  It is a contradiction if the
inequality is strict, so treat it as an equality.  But this will cause
a chain of iterations all landing in the final case of the definition of $\mfh^*$.
Since the same halfspace-forcing counterexample is used, this chain fits exactly
with the conditions of \Cref{fact:antiforce}, and thus some eventual iteration
$t' > t$ will satisfy $\scI_S(\phi_{t'}) < \scI_S(\phi_t)$, a contradiction.
Since $t$ was the least counterexample, there are no counterexamples, and the result
follows.

Finally, suppose the extra condition that every
halfspace-forcing counterexample $(\phi,\psi,H)$
has $H^c$ 1-forcible by $\cY$.  Then the stronger guarantee of \Cref{fact:antiforce}
is at play, and $\phi_t \not \in S_{\delta(S)}$
can be guaranteed regardless of the choice of $(x_t)_{t\geq 1}$.
\end{proof}

\subsection{Approachability and Minimax Theory.}

Combining the results so far grants a proof of Blackwell's Approachability Theorem
in the general (nonconvex) form of Hou~\cite{hou_approach} and Spinat~\cite{spinat_approach},
but in the general setting of this manuscript, with no reliance on minimax structure.

\begin{theorem}
\label{fact:blackwell_approach}
Let $S\subseteq \R^d$ be given.
Then $S$ is approachable iff $S$ contains a compact A-set.
\end{theorem}
\begin{proof}
Without loss of generality, $S$ may be presumed compact.  If $S$
contains an A-set, then \Cref{fact:approach_sufficient} provides approachability.
If $S$ does not contain an A-set, then $S_\infty\subseteq S$ is not an A-set,
and \Cref{fact:shrinkage_operator_iteration} provides non-approachability.
\end{proof}

The next question then, is what is gained by minimax structure?
The answer goes back to \Cref{fact:force_minimax}: minimax structure allows
implies that 2-forcible halfspaces are also 1-forcible.  This effect propagates
to approach games: minimax structure makes the player order inconsequential.

\begin{theorem}
\label{fact:approach_avoid}
Suppose $(\cX,\cY,\zcf)$ has minimax structure, and let $S\subseteq \R^d$ be given.
Exactly one of the following statements holds:
\begin{enumerate}
\item
$S$ is approachable with strategy $\mfg^*$;
\item
$S$ is avoidable with strategy $\mfh^*$.
\end{enumerate}
\end{theorem}
\begin{proof}
Suppose $S$ contains an A-set: by \Cref{fact:approach_sufficient}, it is approachable
with strategy $\mfg^*$.

Now suppose $S$ does not contain an A-set;
by \Cref{fact:shrinkage_operator_iteration}, $S_\infty \subseteq S$ not an A-set
means $T(S) < \infty$ and \Cref{fact:mfh_star} can be applied.  But minimax structure,
combined with \Cref{fact:force_minimax}, grants that every halfspace-forcing
counterexample $(\phi,\psi,H)$ has $H^c$ 1-forcible by $\cY$.  Thus the stronger
guarantee of \Cref{fact:mfh_star} holds, and $\mfh^*$ is an avoidance strategy.
%
%
\end{proof}

Note that without minimax structure, there exists sets which are neither avoidable
nor approachable.  It suffices to consider a game in $\R^1$ with 
$a := \inf_x \sup_y \zcf(x,y) > \sup_y\inf_x \zcf(x,y) =: b$, just 
as in \Cref{sec:forcing}.  In particular, the set $H := (-\infty, (a+b) / 2]$ is
neither approachable nor avoidable.


\section{Stochastic Games.}
\label{sec:random}
The final missing piece is stochasticity; throughout this section, take
$\cX$ and $\cY$ to be families of distributions.  For convenience, given
a pair of distributions $(\mu,\nu)\in\cX\times\cY$, 
let $\sfE \zcf$ denote the map
\[
(\mu,\nu) \mapsto \sfE_{X\sim \mu, Y\sim \nu}(\zcf(X,Y)).
\]
This setting is nearly an analog of the deterministic game: each player has
access to the history of all distributions chosen so far, but additionally
the sampled payoff $\zcf(X_t,Y_t)$ where $(X_t,Y_t) \sim (\mu_t,\nu_t)$.
Accordingly, strategies are now families of maps from these augmented histories to
members of $\cX$ and $\cY$.  There is some disagreement between authors
as to the exact contents of the history; the history here was also used by
Hou~\cite{hou_approach}.

\begin{definition}
A set $S$ (or, for clarity, a tuple $(\cX,\cY,\zcf,S)$) is
\emph{stochastic approachable} if
\[
\exists \mfg
\centerdot
\forall \epsilon > 0
\centerdot
\exists T
\centerdot
\forall \mfh
\centerdot
\sfP(\exists t\geq T\centerdot \phi_t \not\in S_\epsilon) \leq \epsilon,
\]
where the probability $\sfP$ is taken over the player history product distribution.
\end{definition}

The crucial result is that the family of approachable sets is effectively the
same, with and without randomness.  Moreover, the method of proof reveals that
the deterministic strategies $\mfg^*$ and $\mfh^*$ can be simply adapted to handle
the stochastic case.


\begin{theorem}
\label{fact:derandomize}
Let $\cX$ and $\cY$ be families of distributions over a pair of sets $\bbX$ and
$\bbY$, and $\zcf:\bbX\times \bbY \to \R^d$ and 
$\sfE \zcf : \cX\times\cY \to \R^d$ are as above.
For any set $S\subseteq \R^d$,
$(\cX,\cY,\zcf,S)$ is stochastic approachable
iff
$(\cX,\cY,\sfE \zcf,S)$ is approachable.
\end{theorem}

When adapting $\mfg^*$ and $\mfh^*$ to the stochastic case, the following concentration
result will be used.  Although the independence statement may seem too strenuous,
it will suffice because the adapted strategies will never care about the sampled values
$X_t\sim\mu_t$ and $Y_t\sim\nu_t$, but rather only use the source
distributions $\mu_t$ and $\nu_t$.

\begin{lemma}
\label{fact:blackwell_slln}
Let
any $\epsilon > 0$ and any sequence 
$(Z_i)_{i=1}^\infty$ of $d$-dimensional independent random variables
be given with $\|Z_i\|\leq \gamma$ almost surely.
Define $S_n := n^{-1}\sum_{i=1}^n Z_i$.
Then for any 
$N \geq \frac {\gamma^2}{2\epsilon^2}\ln\left(
\frac {2d}{\epsilon(1-\exp(-\epsilon^2/(2\gamma^2)))}
\right)$,
\[
\sfP\left(\exists n\geq N\centerdot  \|S_n - \sfE(S_n)\| \geq \epsilon\right)
\leq \epsilon.
\]
\end{lemma}
\begin{proof}
First consider any fixed $n \geq N$.
By norm equivalence, $\|S_n - \sfE(S_n)\| \leq \|S_n - \sfE(S_n)\|_1$.
Applying Hoeffding's inequality to any fixed coordinate $j$ (which can vary by
at most $2\gamma$),
\[
\sfP(|(S_n)_j - (\sfE(S_n))_j| \geq \epsilon)
\leq 2\exp(-2n\epsilon^2/(2\gamma)^2).
\]
Unioning these events for all coordinates and all $n\geq N$, it follows that
\begin{align*}
\sfP\left(
\exists n\geq N\centerdot \|S_n - \sfE(S_n)\| \geq \epsilon
\right)
&\leq \sum_{n\geq N}
\sfP\left(
\|S_n - \sfE(S_n)\|_1 \geq \epsilon
\right)
\\
&\leq 2d \exp(-2N \epsilon^2/\gamma^2) \sum_{i\geq 0} \exp(-i\epsilon^2/(2\gamma^2))
\\
&\leq \epsilon(1-\exp(-2\epsilon^2/\gamma^2)) \sum_{i\geq 0} \left(
\exp(-\epsilon^2/(2\gamma^2))
\right)^i.
\end{align*}
\end{proof}

\ifarxiv
\begin{proof}[\Cref{fact:derandomize}]
\else
\begin{proofof}{\Cref{fact:derandomize}}
\fi
($\Longleftarrow$)
Suppose $(\cX,\cY,\sfE \zcf,S)$ is approachable, and let $\epsilon > 0$ be given.
Define a stochastic approach strategy $\mfg$ as follows.
In every iteration,
$\mfg$ invokes $\mfg^*$ (with schedule $\tau_t = \gamma 2^{-t}$).
Since $S$ is approachable, by \Cref{fact:blackwell_approach} it is an A-set;
thus, for $t \geq T_1 := \lceil 12\gamma^2/\epsilon^2\rceil$,
$\sfE(\phi_t) \in S_{\epsilon/2}$.
Next, even though there may be dependence between $X_t$ and $Y_t$ across iterations, 
they are
independent given their distributions $\mu_t$ and $\nu_t$.
As such, \Cref{fact:blackwell_slln} may be applied,
and
there exists $T_2$ so that
$\sfP\left(\exists t\geq T_2\centerdot \|\phi_t - \sfE(\phi_t)\| \geq \epsilon/2\right) \leq \epsilon/2$.
Taking $T := \max\{T_1,T_2\}$,
\[
\sfP\left(\exists t \geq T\centerdot \phi_t \not \in S_\epsilon\right)
\leq
\sfP\left(\exists t \geq T\centerdot\phi_t \not \in B(\sfE(\phi_t),\epsilon/2)\right)
+ \sfP\left(\exists t \geq T\centerdot \sfE(\phi_t) \not \in S_{\epsilon/2}\right)
\leq \epsilon.
\]

($\Longrightarrow$)
This direction is established via contrapositive: suppose $(\cX,\cY,\sfE \zcf, S)$
is not approachable,
and construct an opponent strategy $\mfh$ by once again feeding the distribution history
to a deterministic strategy, this time $\mfh^*$.
By \Cref{fact:mfh_star}, there exist $T(S)$ and $\delta' > 0$ so that
$\sfE(\phi_t) \not \in S_{\delta'}$ for infinitely many $t$;
set $\delta := \min\{ \delta', 1\}$.
Invoking \Cref{fact:blackwell_slln} (with the same independence considerations as 
for the converse),
there exists $T'$ such that $\sfP(\forall t\geq T'\centerdot \phi_t \in B(\sfE(\phi_t),
\delta /2)) \geq 1-\delta/2$, and thus, for any $T \geq \max\{T',T(S)\}$,
\[
\sfP(\exists t \geq T\centerdot \phi_t \not \in S_{\delta /2}) \geq 1- \delta/2.
\]
\ifarxiv
\end{proof}
\else
\end{proofof}
\fi

\begin{remark}
Note $\sfE \zcf$ is bilinear, and suppose $\cX$ and $\cY$ are compact convex.
This grants minimax structure (cf. \Cref{sec:minimax_families}), and one may view the
stochasticity as an operator embedding a tricky game into a more structured setting.
\end{remark}

\appendix
\section{An Example.}
\label{sec:example}
Many beautiful approachability examples can be found elsewhere.  Here are two
favorites:
\begin{itemize}
\item Blackwell~\cite{blackwell_approach} presents a game which is neither approachable nor excludable;
\item Spinat~\cite{spinat_approach} presents an approachable
nonconvex set containing no smaller approachable set.
\end{itemize}
This section presents a very simple game to demonstrate the relationship of 
1-forcing, 2-forcing, and approachability.  For its duration, fix
\begin{align*}
\cX := \cY &:= \{\left[\begin{smallmatrix}\alpha \\ 1-\alpha\end{smallmatrix}\right]
: \alpha \in [0,1]\},
\\
\zcf(x,y) &:= \begin{bmatrix}
x_1y_1\\x_2y_2
\end{bmatrix}.
\end{align*}
Since this game is symmetric, it suffices to consider the perspective of one player.
Notice that the minimal 1-forcible sets are precisely
\[
L_x := \zcf(x,\cY)
= \left\{
\left[\begin{smallmatrix}\alpha x_1 \\ 0\end{smallmatrix}\right]
+ \left[\begin{smallmatrix} 0 \\ (1-\alpha)x_2\end{smallmatrix}\right]
: \alpha \in [0,1]
\right\};
\]
these are minimal in the sense that 
every 1-forcible set must contain $L_x$ for some $x$, and each $L_x$ has no
1-forcible proper subsets.

For a convex approachable set that is not 1-forcible, consider
\[
S_0 := \{(\alpha,\alpha) : \alpha \in [0,1/2]\}.
\]
Since every halfspace containing $S_0$ also contains some $L_x$,
\Cref{fact:blackwell_approach} grants approachability.

This game is a tensor of order 3, and the results of \Cref{sec:minimax_families}
grant it has 
the minimax property. But there exist 2-forcible sets which are not
halfspaces, thus not covered by \Cref{fact:force_minimax}, and are not approachable.
Consider in particular the set 
\[
S_1 := \big([1/2,1] \times \{0\}\big) \cup \big(\{0\} \times (1/2,1]\big).
\]
Since $S_1$ intersects every set $L_x$, it follows that it is 2-forcible;
by the aforementioned symmetry of $\zcf$, it is 2-forcible for $\cX$ and for $\cY$.
On the other hand since $S_0$ is approachable but
$\rho(S_0,S_1) > 0$, the $\cY$-player can play $\mfg^*$ to approach $S_0$,
meaning $S_1$ is not approachable.

Finally, as per \Cref{rem:tricky_superset_approach}, there exist tricky supersets
of A-sets for which it is not clear how to construct an approach strategy (without 
resorting to determining an inner A-set, and invoking $\mfg^*$ on it).  In particular,
choose any $x_1,x_2\in\cX$ with $x_1\neq x_2$, and construct $S_2$ by removing a small
piece from $L_{x_2}$ in $L_{x_1} \cup L_{x_2}$.
Any approach strategy for $S_2$ must know to avoid $L_{x_2}$; otherwise,
the $\cY$-player could let arbitrarily many $\phi_t$ stay inside the disconnected 
piece of $L_{x_2}$, and then begin forcing points inside $L_{x_1}$, thus taking $\phi_t$
outside of $S_2$.  Since approachability requires a uniform bound over all opponent
strategies, this means the provided strategy for $\cX$ is not an approach strategy.
Thus, it is necessary for the $\cX$-player strategy to knowingly avoid this
disconnected piece of $L_{x_2}$.

\section{Geometric Facts.}
This section collects a few geometric facts requiring careful proof.
\subsection{Proof of \Cref{fact:antiforce}.}
\label{sec:antiforce_proof}
Throughout this section,
suppose $(\phi,\psi,H)$ is merely a halfspace-forcing candidate; 1-forcing 
of $H$ by $\cX$ or $\cY$ will only come into play at the end.  As
in \Cref{fact:antiforce}, set $\tau := \rho(\psi,H^c) > 0$,
and set $H' := \overline{H^c}$.


\begin{lemma}
\label{fact:af2}
Set $\epsilon_0 := \tau(1 - \gamma/\sqrt{\gamma^2+\tau^2})$,
and $\epsilon_1 := \tau\epsilon_0 / (2\sqrt{4\gamma^2+\tau^2})$.
For any $z\in H'\cap \zcf(\cX,\cY)$,
there exists $\eta \in (0,1)$ such that
$B(\eta \psi + (1-\eta)z, \epsilon_1) \subseteq B(\phi,\rho(\phi,\psi))$
and $\rho(\eta \psi + (1-\eta)z,H') \geq \epsilon_1$.
\end{lemma}
\begin{proof}
First,
it suffices to consider $z$ on the boundary of $H'$: given $z' \in (H')^o$ with $z$ denoting the intersection
of $[\psi,z']$ with the boundary of $H$, the desired $\eta$ for $z$ can be converted
into $\eta' \in (0,1)$ for $z'$ satisfying $\eta'\psi + (1-\eta')z' = \eta\psi + (1-\eta)z$.

Second,
discard the scenario that $\{\psi,z,\phi\}$ are collinear:
since $\epsilon_1 < \tau/2 \leq \rho(\phi,\psi)/2$,
taking $\eta$ close to $1/2$ suffices.

Third,
it suffices to exhibit an $\eta \in[0,1)$ satisfying the single 
property $B(\eta \psi + (1- \eta)z, \epsilon_0) \subseteq B(\phi,\rho(\phi,\psi))$.
This $\eta$ may potentially violate the second condition above:
perhaps
$\rho(\eta \psi + (1-\eta)z,H') < \epsilon_1$.
To see how this can be adjusted,
set 
$z' := \eta\psi + (1-\eta)z$
and $\phi'$ to be the projection of $z'$ onto 
$[\phi,\psi]$; thus
the right triangle $\{\psi,\phi',z'\}$ has short sides with lengths
$\rho(\psi,\phi') \geq \tau-\epsilon_1 \geq \tau/2$ and
$\rho(\phi',z') \leq \gamma$.  By similarity of triangles, the point $z''$
along $[z',\psi]$ which is $\epsilon_0/2$ away from $z'$ must satisfy
\begin{align*}
\rho(z'',H^c) 
&\geq \rho(z'',[z',\phi'])
= \rho(z'',z')\left(\frac{\rho(\phi',\psi)}{\rho(z',\psi)}\right)
= \frac {\epsilon_0}{2}\left(\frac{1}{\sqrt{1 + (\rho(\phi',z')/\rho(\phi',\psi))^2}}\right)
\\
&\geq \frac {\epsilon_0}{2}\left(\frac{\tau}{\sqrt{\tau^2 + 4\gamma^2}}\right) 
=\epsilon_1.
\end{align*}
Meanwhile,
$
B(z'', \epsilon_1) \subseteq B(z'',\epsilon_0/2) \subseteq B(z', \epsilon_0),
$
meaning both properties are satisfied.  
And since $z'' \in (\psi,z)$, this grants an $\eta\in(0,1)$ with all desired properties.

The remainder of the proof will be divided into the two cases, as in \Cref{fig:af2},
whether $\angle(\phi,z,\psi) < \pi/2$ or not.  Set 
$\phi^*$ to be the intersection of $[\phi,\psi]$ with the boundary of $H$,
meaning $\rho(\psi,\phi^*) = \tau$.


\begin{figure}[]
\centering
\subfloat[$\angle(\phi,z,\psi) < \pi/2$.]{
\includegraphics[width=0.4\textwidth]{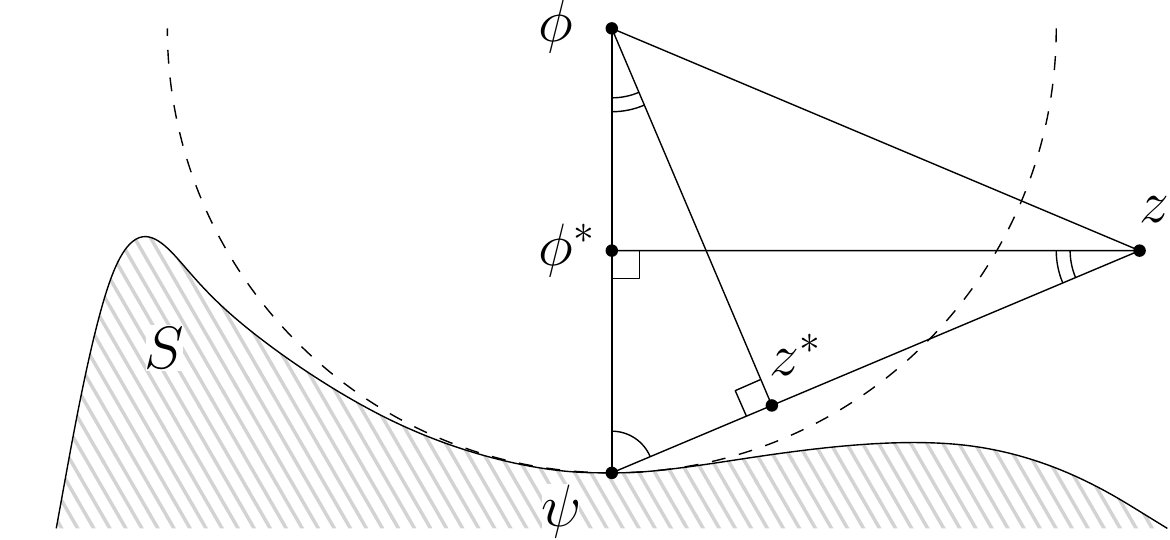}
\label{fig:af2:a}
}
\hspace{0.1\textwidth}
\subfloat[$\angle(\phi,z,\psi) \geq \pi/2$.]{
\includegraphics[width=0.4\textwidth]{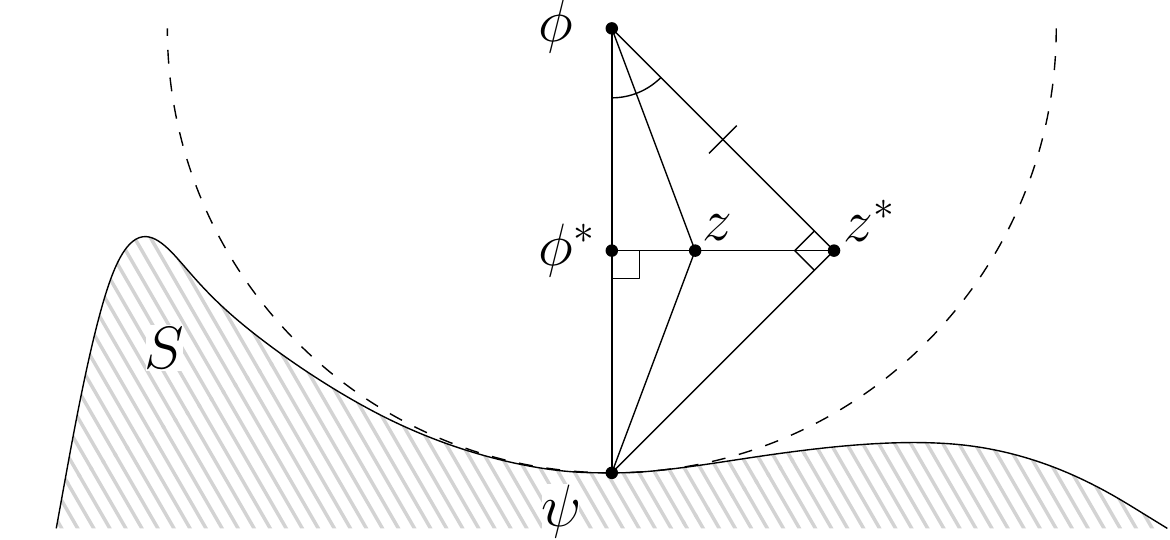}
\label{fig:af2:b}
}
\caption{Two cases in the proof of \Cref{fact:af2}.}
\label{fig:af2}
\end{figure}

Suppose $\angle(\phi,z,\psi) < \pi/2$, and designate $z^*$ as the
point along $[z,\psi]$ satisfying $[z^*,\psi] \perp [z^*,\phi]$,
as in \Cref{fig:af2:a}.  This $z^*$ will provide the eventual $\eta$.
By similarity of triangles $\{\psi,z^*,\phi\}$ and $\{\psi,\phi^*,z\}$,
\begin{align*}
\frac{\rho(z^*,\phi)}{\rho(\psi,\phi)}
&= \frac {\rho(z,\phi^*)}
{\sqrt{\rho(z,\phi^*)^2 + \rho(\psi,\phi^*)^2}}.
\end{align*}
Rearranging and 
making use of $\rho(\phi^*,z) \leq \gamma$
and $\rho(\psi,\phi) \geq \rho(\psi,\phi^*) = \tau$,
the desired inequality is
\begin{align*}
\rho(\phi,\psi) - \rho(z^*,\phi)
&
= \rho(\phi,\psi)
\left(
1 - \frac {1}{\sqrt{1 + \rho(\psi,\phi^*)^2/\rho(z,\phi^*)^2}}
\right)
\geq \rho(\phi,\psi)
\left(
1 - \frac {1}{\sqrt{1 + \tau^2/\gamma^2}}
\right)
\\
&\geq \tau
\left(
1 - \frac {\gamma}{\sqrt{\gamma^2 + \tau^2}}
\right)
= \epsilon_0.
\end{align*}

Finally, consider $\angle(\phi,z,\psi) \geq \pi/2$.  This time,
place $z^*$ along the line through $\{z,\phi^*\}$ so that
$[z^*,\phi] \perp [z^*,\psi]$ as in \Cref{fig:af2:b}.
By construction, $\rho(z^*,\phi) \geq \rho(z,\phi)$,
thus it suffices to upper bound $\rho(z^*,\phi)$, and $z$ will be the desired 
point, meaning $\eta = 0$.
By similar of triangles $\{\psi,\phi,z^*\}$
and $\{z^*,\phi,\phi^*\}$,
\[
\frac {\rho(\psi,\phi)}{\rho(\phi,z^*)}
= \frac {\rho(z^*,\phi)}{\rho(\phi,\phi^*)}.
\]
Using $\rho(\psi,\phi) \geq \rho(\psi,\phi^*) = \tau$
and $\rho(\phi^*,\phi)\leq \gamma$,
\begin{align*}
\rho(z^*,\phi)
&= \rho(\psi,\phi) \sqrt{\frac{\rho(\phi,\phi^*)}{\rho(\phi,\phi^*) + \rho(\phi^*,\psi)}}
= \rho(\psi,\phi) \sqrt{\frac{1}{1 + \rho(\phi^*,\psi)/\rho(\phi,\phi^*)}}
\leq \rho(\psi,\phi) \sqrt{\frac{1}{1 + (\tau/\gamma)^2}},
\end{align*}
with the remainder as before.
\end{proof}

\begin{lemma}
\label{fact:af3}
Set $\epsilon_2 := \epsilon_1/2$.
For any $p\in B(\psi, \epsilon_2)$,
and any $z\in H'\cap \zcf(\cX,\cY)$,
there exists $\eta \in (0,1)$ such that
$B(\eta p + (1-\eta)z, \epsilon_2) \subseteq B(\phi,\rho(\phi,\psi))$
and $\rho(\eta p + (1-\eta)z,H') \geq \epsilon_2$.
\end{lemma}
\begin{proof}
Choosing the $\eta$ granted by \Cref{fact:af2},
\begin{align*}
\rho(\eta p + (1-\eta)z, \phi)
=
\|\eta (p -\psi + \psi) + (1-\eta)z -  \phi\|
\leq 
\eta \underbrace{\| p -\psi\|}_{\leq \epsilon_2}
+
\underbrace{\| \eta \psi + (1-\eta)z-  \phi\|}_{\leq \rho(\phi,\psi) - \epsilon_1}
\leq
\rho(\phi,\psi) - \epsilon_2.
\end{align*}
For the other property,
\begin{align*}
\rho(\eta p + (1-\eta)z, H')
&=
\inf_{q\in H'} \| \eta p + (1-\eta)z - q\| + \|\eta(\psi - p)\| - \|\eta(\psi-p)\|
\geq
\rho(\eta \psi + (1-\eta)z, H') - \epsilon_2.
\end{align*}
\end{proof}

\begin{lemma}
\label{fact:af4}
Set $\epsilon_3 := \epsilon_2/2 = \epsilon$.
Let any
$p\in B(\psi, \epsilon_2)$,
$T \geq \lceil 8\gamma/\epsilon_3\rceil$,
and $(z_i)_{i=1}^N \in (H'\cap \zcf(\cX,\cY))^N$ with $N = \lceil T\gamma\epsilon_3/8\rceil$
be given.
Then there exists an integer $M\leq N$ such that
\[
\rho\left(
\frac {Tp + \sum_{i=1}^M z_i}{T+M}, \phi
\right) \leq \rho(\phi,\psi) - \epsilon_3.
\]
\end{lemma}
\begin{proof}
For every $z\in H'\cap f(\cX,\cY)$, let $\eta_z\in (0,1)$ be the value granted by
\Cref{fact:af3} so that 
\[
B(\eta_z p + (1-\eta_z) z, \epsilon_2) \subseteq B(\phi,\rho(\phi,\psi)) \cap H.
\]
Next, set $U_z := B(\eta_z p +(1-\eta_z) z, \epsilon_3)$,
and $\cU$ to be the convex hull of the union of all $U_z$;
it follows that $\cU\subseteq B(\phi,\rho(\phi,\psi) - \epsilon_2)$ and
$\rho(\cU,H^c) = \inf_{u\in \cU} \rho(u,H^c) \geq \epsilon_2$.

For every $i\leq N$, consider the partial averages
\[
q_i : = \frac {Tp + \sum_{j=1}^i z_j}{T+i},
\]
and take $w_i$ to denote the point on the boundary of $H'$ so that
$q_i \in [p,w_i]$ (or $w_i = q_i$ if $q_i\in H'$).  Correspondingly,
choose $\eta_i$ as  provided by \Cref{fact:af3} so that
$\eta_i p + (1-\eta_i)w_i \in \cU$,
and set $\mu_i\in[0,1]$ so that $q_i = \mu_i p + (1-\mu_i) w_i$. Since
$N \geq T\gamma\epsilon_3/8$ and $\gamma>0$,
\[
\left\|q_N - \frac 1 N \sum_{i=1}^Nz_i\right\|
= \frac {1}{T+N}\left\|Tp + \sum_{i=1}^N z_i  - \frac {T+N} N \sum_{i=1}^Nz_i\right\|
= \frac {T}{T+N} \left\|p - \frac 1 N\sum_{i=1}^N z_i\right\|
\leq \frac {T\gamma}{T+N}
\leq \epsilon_3/8,
\]
meaning in particular that $\rho(q_N,H') \leq  \epsilon_3/8$,
so has gotten beyond $\cU$,
thus $\mu_N \geq \eta_N$.  The final step will be to show $(q_i)_{i=1}^N$ must have
actually passed through $\cU$.

Now let $k$ be the first index such that $\mu_k\geq \eta_k$,
meaning $\mu_{k-1} < \eta_{k-1}$.  Suppose contradictorily
that $q_k\not \in \cU$ and $q_{k-1}\not\in\cU$.  Since $T \geq 8\gamma/\epsilon_3$,
it follows that $\|q_k-q_{k-1}\| \leq \epsilon_3/8$.
This implies that
the ball centered at $\eta_i p + (1-\eta_i)w_k$ with radius $\epsilon_3/4$
(which is within $B(\eta_i p + (1-\eta_i)w_k,\epsilon_3)\subseteq \cU$) contains
a point along the line $[p,q_{k-1}]$, and as such $q_{k-1}\in \cU$, a contradiction.
Thus the desired $M$ exists.
\end{proof}

\ifarxiv
\begin{proof}[\Cref{fact:antiforce}]
\else
\begin{proofof}{\Cref{fact:antiforce}}
\fi
It is given that $(\phi,\psi,H)$ is a halfspace-forcing example,
by which it follows that $H^c$ and hence $H'$ can be 2-forced by $\cY$.
Given a sequence $(x_i)_{i=1}^N$, this grants $(y_i)_{i=1}^N$ with
$\zcf(x_i,y_i) \in H'\cap \zcf(\cX,\cY)$, whereby \Cref{fact:af4}
may be applied, and the result follows.

Now suppose the stronger condition that $H^c$ (and $H'$) can be 1-forced by $\cY$.
Thus there exists a single $\bar y$ so that the choice $y_i =\bar y$ grants
$\zcf(x_i,y_i)\in H' \cap \zcf(\cX,\cY)$, whatever the choice of $x_i\in \cX$.
Once again, \Cref{fact:af4} may be applied.
\ifarxiv
\end{proof}
\else
\end{proofof}
\fi

\subsection{Halfspace-forcing Counterexamples of Similar Sets.}
This subsection will use the existence of a halfspace-forcing counterexample 
$(\phi,\psi,H)$ on
a set $S$ to produce another counterexample on some similar set $S'$.  
It will be supposed that $\rho(\phi,\psi) = \rho(\psi,H^c) =: \tau$;
this comes without loss of generality,
since otherwise $\rho(\phi,\psi)>\tau$, in which case some other $\phi'$ may be chosen along the
segment $[\phi,\psi]$, but closer to $\psi$, and still forming a halfspace-forcing 
counterexample $(\phi',\psi,H)$.

To this end, given any pair $(\phi,q)$, define $H(\phi,q)$ to be the halfspace with
normal $q-\phi$ having $\phi$ on its boundary; that is,
\begin{equation}
H(\phi,q) := \{ z\in\R^d : \ip{z}{q-\phi} \leq \ip{\phi}{q-\phi} \}.
\label{eq:defn:H}
\end{equation}
For instance, after the adjustment placing $\phi$ on the boundary of $H$,
$H = \overline{H(\phi,\psi)^c}$,
and $(\phi,\psi,\overline{H(\phi,\psi)^c})$ is a halfspace-forcing counterexample.

For any pair $(\phi,q)$ and $\delta\in(0,\rho(\phi,q))$, let $R_{\delta}(\phi,q)$
be a shell of radius $\rho(\phi,q)$ and width $\delta$ around $\phi$:
\[
R_\delta(\phi,q) := B(\phi,\rho(\phi,q)) \setminus B(\phi,\rho(\phi,q)-\delta)^o
= \left\{
z\in\R^d : \rho(\phi,q) \geq \rho(\phi,z) \geq \rho(\phi,q)-\delta
\right\}.
\]

\begin{lemma}
\label{fact:nested_balls:helper}
Let $\phi,\psi$ be given where $H(\phi,\psi)$ can be 2-forced by $\cY$,
and set $\tau := \rho(\phi,\psi)$
and $\phi' := (\phi + \psi)/2$.
Then there exists $\delta >0$ so that every
$q\in B(\phi',\tau/2) \cap R_\delta$ satisfies
\begin{enumerate}
\item $\rho(q,\psi) \leq \tau/4$,
\item $H(\phi',q)$ can be 2-forced by $\cY$.
\end{enumerate}
\end{lemma}
Many of the relevant quantities appear in \Cref{fig:af4}.
\begin{figure}[]
\centering
\includegraphics[width=0.5\textwidth]{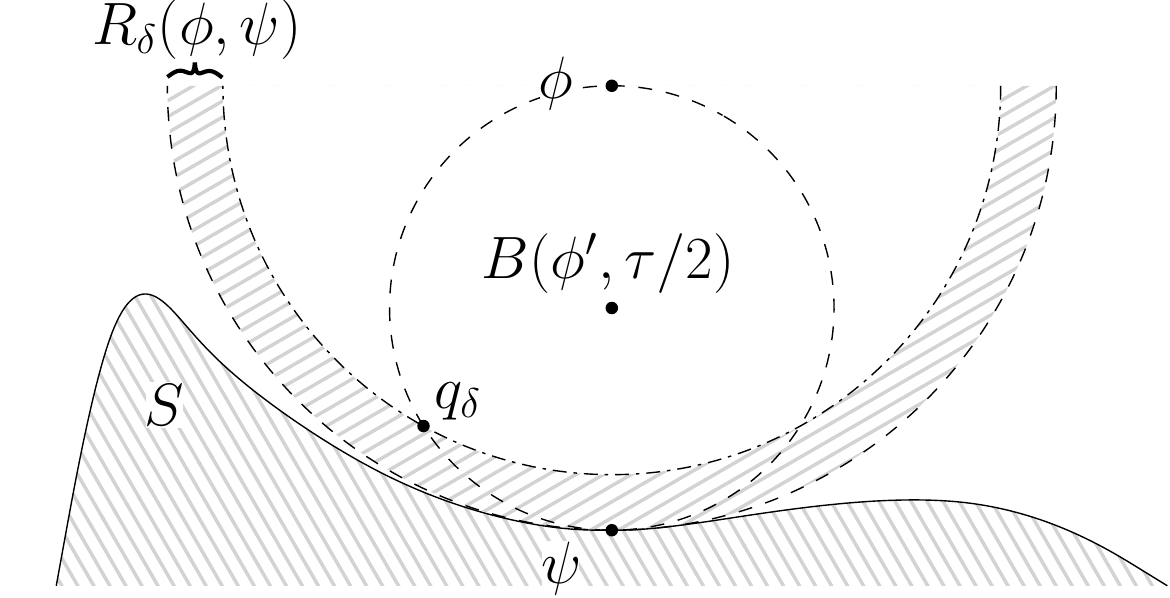}
\caption{Objects in the proof of \Cref{fact:nested_balls:helper}.}
\label{fig:af4}
\end{figure}
\begin{proof}
For any $\delta\in (0,\tau/2)$,
consider two maximization problems, where capitalization refers to the 
presence of multiple optima:
\begin{align*}
q_\delta
&\in \Argmax \{\angle(\psi,\phi,q) : q\in B(\phi',\tau/2)\cap R_\delta\},
\\
q_\delta'
&\in \Argmax \{\rho(\psi,q) : q\in B(\phi',\tau/2)\cap R_\delta\}.
\end{align*}
By rotational symmetry of all quantities around $[\phi,\psi]$, it suffices
to consider $q_\delta,q_\delta'$ which lie in the same plane, and lie on or
above $[\phi,\psi]$.  In this case, $q_\delta = q_\delta'$: this follows since 
the farthest point from $\psi$ must lie on the inner edge of $R_\delta$,
and it must be as high above $[\phi,\psi]$ as possible.

As such, to establish the first desired property, note that $\rho(\psi,q_\delta)$ decreases
continuously as $\delta\downarrow 0$, and thus there exists $\delta_1$ so that
$\rho(\psi,q_{\delta_1}) \leq \tau/4$. 
Since $q_\delta$ is the 
farthest point, the property follows.
%
%
%

For the second property, since $H^c$ can be 2-forced by $\cY$, it follows that
$H^c\cap \zcf(\cX,\cY)$ can be 2-forced by $\cY$.  Now consider any point 
$q\in B(\phi,\tau/2) \cap R_\delta$, and consider the halfspace
$H(\phi',q)$ as defined above.
Once again, it suffices by rotational  symmetry of the relevant quantities around
$[\phi,\psi]$ to consider $q$ as lying in a plane above $[\psi,\phi']$ and
$H(\phi',q)$ appearing as a line within this plane.
If $H(\phi',q)$ does not cross $H^c\cap
\zcf(\cX,\cY)$ within this plane, then $H(\phi',q)$ can be 2-forced by $\cY$.

To this end, 
since $\angle(\psi, \phi', q_\delta) = \max \{ \angle(\psi,\phi,q) : q\in B(\phi',\tau/2)
\cap R_\delta\}$, it suffice to prove that $H(\phi',q_\delta)$ can be 2-forced.  But as
$\delta\downarrow 0$, the preceding analysis grants
$q_\delta$ becomes arbitrarily close to $\psi$,
thus the normal of $q_\delta - \phi'$ of $H(\phi',q)$ becomes increasingly close to
that of $H(\phi,\psi)$, whereas they intersect $[\phi,\psi]$ at points $\tau/2$ 
apart;
so there 
must exist a $\delta_2$ sufficiently small.
%
%
%
%

Setting $\delta := \min\{\delta_1,\delta_2\}$, the result follows.
\end{proof}

\begin{proposition}
\label{fact:nested_balls_redone}
Let $S\subseteq\R^d$ be given with halfspace-forcing counterexample
$(\phi,\psi,H)$, and set $\tau := \rho(\phi,\psi)$.  Then there exists $\delta > 0$ so
that every $S'\subseteq \R^d$ satisfying $\Delta(S,S')<\delta$
has a halfspace-forcing counterexample $(\phi_s', q', \overline{H(\phi'_s,q')^c})$ with
$\rho(H(\phi'_s,q'), q') \geq \tau/4$, where $H(\cdot,\cdot)$ is as in 
\Cref{eq:defn:H}.
\end{proposition}
\begin{proof}
Let $\delta_0>0$ be as provided by \Cref{fact:nested_balls:helper},
set $\delta := \delta_0/2$,
and let $S'$ be any set satisfying $\Delta(S,S') < \delta_0/2$.
Also following \Cref{fact:nested_balls:helper}, set 
$\phi' := (\psi+\phi)/2$.
Define 
\begin{align*}
\phi_s 
&:= \phi - \frac {\delta_0} 2 \left(
\frac {\phi - \psi}{\|\phi-\psi\|}
\right),
\\
\phi'_s 
&:= \phi' - \frac {\delta_0} 2 \left(
\frac {\phi - \psi}{\|\phi-\psi\|}
\right),
\\
\psi_s 
&:= \psi - \frac {\delta_0} 2 \left(
\frac {\phi - \psi}{\|\phi-\psi\|}
\right);
\end{align*}
these are just copies of $\phi,\phi',\psi$ shifted by $\delta_0/2$ along the direction
to $\psi$ from $\phi$.
As such, there is a bijection (via this shift) granting
\[
B(\phi',\tau/2)\cap R_{\delta_0}(\phi,\psi)
\ni q \mapsto
q_s
\in
B(\phi'_s,\tau/2)\cap R_{\delta_0}(\phi_s,\psi_s).
\]
It follows from \Cref{fact:nested_balls:helper} that every such $q_s$ satisfies 
$q_s \in B(\psi_s,\tau/4)$, and $H(\phi_s',q_s)$, which 
contains $H(\phi',q)$, is 2-forcible by $\cY$.


Since $B(\phi,\tau)^o \cap S = \emptyset$, the triangle inequality grants
\begin{align*}
B(\phi_s,\tau - \delta_0/2)^o \cap S &= \emptyset,\\
B(\phi_s,\tau - \delta_0)^o \cap S' &= \emptyset.
\end{align*}
Combining this with the definition of $R_{\delta_0}(\phi_s,\psi_s)$,
\begin{align*}
S \cap B(\phi_s, \tau) 
&\subseteq R_{\delta_0}(\phi_s,\psi_s),
\\
S' \cap B(\phi_s, \tau) 
&\subseteq R_{\delta_0}(\phi_s,\psi_s).
\end{align*}
So by the choice of $\delta_0$ and guarantees of \Cref{fact:nested_balls:helper} 
(points $q_s$ in this cap satisfy $q_s \in B(\psi_s,\tau/4)$),
\begin{align*}
S \cap B(\phi_s', \tau/2) 
\quad\subseteq\quad B(\psi_s,\tau/4) \cap B(\phi_s', \tau/2) \cap R_{\delta_0}(\phi_s,\psi_s)
\quad\subseteq\quad B(\psi_s,\tau/4) \cap R_{\delta_0}(\phi_s,\psi_s),
\\
S' \cap B(\phi_s', \tau/2) 
\quad\subseteq\quad B(\psi_s,\tau/4)  \cap B(\phi_s', \tau/2) \cap R_{\delta_0}(\phi_s,\psi_s)
\quad\subseteq\quad B(\psi_s,\tau/4) \cap R_{\delta_0}(\phi_s,\psi_s).
\end{align*}
In particular, if $S' \cap B(\phi_s',\tau/2)$ is nonempty,
then the closest point $q'\in S'$ to $\phi_s'$ will fall within
$B(\psi_s,\tau/4) \cap R_{\delta_0}(\phi_s,\psi_s)$,
and satisfies all desired properties:
$H(\phi_s,q')$ was shown earlier to be 2-forced by $\cY$,
thus $(\phi_s', q', \overline{H(\phi'_s,q')})$
is a halfspace-forcing counterexample for $S'$;
combining 
$\rho(q',\psi_s)\leq \tau/4$ 
with 
$\rho(\phi'_s,\psi_s) = \tau/2$
grants
$\rho(H(\phi'_s,q'),q') \geq \tau/4$.

To this end, note that by construction that
\[
B(\psi,\delta_0/2) \subseteq R_{\delta_0}(\phi_s,\psi_s).
\]
$S'$ satisfies $\Delta(S,S') < \delta_0/2$, meaning there exists
$q'$ inside $S'\cap R_{\delta_0}(\phi_s,\psi_s)$.  The result follows.
\end{proof}

\section{A Limit Property.}
The proof of 
Spinat's~\cite{spinat_approach}
Lemma~1 appears to be incomplete.  The statement is
interesting, so this appendix establishes a strengthened form,
albeit by different means.
\begin{lemma}
\label{fact:limit_property}
Let a sequence of compact approachable sets $(U_i)_{i\geq 1}$ be given.
If $(U_i)_{i\geq 1}$ is convergent in Hausdorff metric, then its 
limit $U$ is a compact approachable set.
\end{lemma}
\begin{proof}
Since each $U_i$ is approachable, \Cref{fact:blackwell_approach} guarantees that
it contains a compact (nonempty) A-set $A_i$.
Completeness of the Hausdorff metric on compact nonempty sets
grants (perhaps by passing to a subsequence)
the $(A_i)_{i\geq 1}$ have a limiting (compact nonempty) set $A$.  

It must be the case that $A\subseteq U$, since otherwise by compactness there
exists $z\in A\setminus U$ with some $\rho(z,U) =: \delta > 0$.
Now consider $j$ (with respect to the above subsequence) 
large enough for $\Delta(U_j,U) \leq \delta /4$
and $\Delta(A_j, A) \leq \delta/2$; it must be the case that $A_j\not \subseteq U_j$,
a contradiction.


Finally, note that $A$ is an A-set.  
Suppose contradictorily that $(\phi,\psi,H)$ is a counterexample to 
halfspace-forcibility with $\tau := \rho(\psi,H^c)$.  By
\Cref{fact:nested_balls_redone}, there must exist a tiny $\delta >0$ 
so that every $A'$ with $\Delta(A,A') < \delta$ has halfspace-forcing counterexamples,
whereas there exists an A-set $A_j$ with $\Delta(A,A_j) < \delta$.

Since $U$ contains an A-set,
\Cref{fact:blackwell_approach} grants that it is approachable.
\end{proof}

\section{Function Families Satisfying the Minimax Property.}
\label{sec:minimax_families}

Recall's Sion's~\cite{sion_1958} minimax theorem, as stated by 
Komiya~\cite{komiya_sion},
with a minor tightening to compact $\cY$ for applicability here.
\begin{theorem}[Sion~\cite{sion_1958}]
Let $\zcg : \cX\times \cY\to \R$ be given where $\cX$ and $\cY$ are compact
convex subsets of linear topological spaces,
$\zcg(\cdot,y)$ is quasiconvex and lower semi-continuous for every $y\in \cY$,
and
$\zcg(x,\cdot)$ is quasiconcave and upper semi-continuous for every $x\in \cX$.
Then 
\[
\min_{x\in\cX}\max_{y\in\cY} \zcg(x,y)
=
\max_{y\in\cY} \min_{x\in\cX}\zcg(x,y);
\]
in particular, each optimization is attainable.
\end{theorem}

Notice that Sion's Theorem may be applied to the case that $\cX$ and $\cY$ are
families of distributions.


\begin{proposition}
Let $f : \cX\times \cY \to \R^d$ be given.
$\ip{f(\cdot,y)}{\lambda}$ is quasiconvex and lower semi-continuous for every $(y,\lambda)\in\cY\times \R^d$ and
$\ip{f(x,\cdot)}{\lambda}$ is quasiconcave and upper semi-continuous for every $(x,\lambda)\in \cX\times\R^d$ iff
$\ip{f(\cdot,\cdot)}{\lambda}$ is continuous and monotonic in each parameter.
\end{proposition}
\begin{proof}
($\Longrightarrow$)
Fix $y$ and $\lambda$; since both $\ip{f(\cdot,y)}{\lambda}$ and $\ip{f(\cdot,y)}{-\lambda}$ are
quasiconvex and lower semi-continuous, it follows that each is continuous, and unrolling
quasiconvexity grants, for any $\alpha \in [0,1]$,
\[
\min\{\ip{f(x_1,y)}{\lambda},\ip{f(x_2,y)}{\lambda}\} 
\leq \ip{f(\alpha x_1 + (1-\alpha)x_2,y)}{\lambda}
\leq \max\{\ip{f(x_1,y)}{\lambda},\ip{f(x_2,y)}{\lambda}\},
\]
which is the statement of monotonicity.  An analogous argument holds for every $x$ and 
$\lambda$.

($\Longleftarrow$)
Continuity implies upper and lower semi-continuity, and every monotonic function $g$
satisfies, for every $x_1,x_2$ and $\alpha\in[0,1]$,
$\min\{g(x_1),g(x_2)\} \leq g(\alpha x_1 + (1-\alpha)x_2) \leq \max\{g(x_1),g(x_2)\}$.
\end{proof}

In the stricter scenario of convexity/concavity, the resulting function family 
is vastly more constrained.

\begin{proposition}
Let $f : \cX\times \cY \to \R^d$ be given.
$\ip{f(\cdot,y)}{\lambda}$ is convex for every $(y,\lambda)\in\cY\times \R^d$ and
$\ip{f(x,\cdot)}{\lambda}$ is concave for every $(x,\lambda)\in \cX\times\R^d$ iff
$f(\cdot,\cdot)$ is affine in each parameter.
\end{proposition}
\begin{proof}
($\Longrightarrow$)
Since for any $\lambda$ and $y$, both $\ip{f(\cdot,y)}{\lambda}$ and $\ip{f(\cdot,y)}{-\lambda}$ are convex,
it follows that $\ip{f(\cdot,y)}{\lambda}$ is affine.  Thus let any $\alpha\in \R$ and $x_1,x_2\in \cX$ be given;
for any $y\in\cY$,
\begin{align*}
f(\alpha x_1 + (1-\alpha) x_2,y)
&= \sum_{i=1}^d \ip{f(\alpha x_1 + (1-\alpha) x_2,y)}{\bfe_i}\bfe_i
\\
&= \sum_{i=1}^d \left(\alpha \ip{f(x_1,y)}{\bfe_i} + (1-\alpha)\ip{f(x_2,y)}{\bfe_i}\right)\bfe_i
\\
&= \alpha f(x_1,y) + (1-\alpha) f(x_2,y).
\end{align*}
Repeating this proof from the perspective of $f(x,\cdot)$, the result follows.

($\Longleftarrow$)
Let $\alpha\in[0,1]$ and $x_1,x_2\in\cX$ be given.  Then, for any $\lambda\in\R^d$,
\[
\ip{f(\alpha x_1 + (1-\alpha) x_2,y)}{\lambda}
=
\ip{\alpha f(x_1,y) + (1-\alpha)f(x_2,y)}{\lambda}
=
\alpha \ip{f(x_1,y)}{\lambda}
+(1-\alpha) \ip{f(x_2,y)}{\lambda}.
\]
Again, the proof for $\cY$ is analogous.
\end{proof}

\ifarxiv
\subsection*{Acknowledgement}
The author thanks his advisor, Sanjoy Dasgupta, for discussions and support.
\else
\acknowledgment{The author thanks his advisor, Sanjoy Dasgupta, for discussions
and support.}
\fi

\bibliographystyle{amsplain}
\bibliography{blackwell}

\providecommand{\bysame}{\leavevmode\hbox to3em{\hrulefill}\thinspace}
\providecommand{\MR}{\relax\ifhmode\unskip\space\fi MR }
\providecommand{\MRhref}[2]{%
  \href{http://www.ams.org/mathscinet-getitem?mr=#1}{#2}
}
\providecommand{\href}[2]{#2}
\begin{thebibliography}{10}

\bibitem{blackwell_approach}
David Blackwell, \emph{An analog of the minimax theorem for vector payoffs},
  Pacific Journal of Mathematics \textbf{6} (1956), no.~1, 1--8.

\bibitem{blackwell_approach_2}
\bysame, \emph{Controlled random walks}, invited address, Institute of
  Mathematical Statistics Meeting, Seattle, Washington, 1956.

\bibitem{borwein_lewis}
Jonathan Borwein and Adrian Lewis, \emph{Convex analysis and nonlinear
  optimization}, Springer Publishing Company, Incorporated, 2000.

\bibitem{foster_vohra}
Dean~P. Foster and Rakesh~V. Vohra, \emph{Asymptotic calibration}, Biometrika
  \textbf{85} (1998), no.~2, 379--390.

\bibitem{HULL}
Jean-Baptiste Hiriart-Urruty and Claude Lemar\'echal, \emph{Fundamentals of
  convex analysis}, Springer Publishing Company, Incorporated, 2001.

\bibitem{hou_approach}
Tien-Fang Hou, \emph{Approachability in a two-person game}, The Annals of
  Mathematical Statistics \textbf{42} (1971), no.~2, 735--744.

\bibitem{komiya_sion}
Hidetoshi Komiya, \emph{Elementary proof for sion's minimax theorem}, Kodai
  Mathematical Journal \textbf{11} (1988), no.~1, 5--7.

\bibitem{munkres_topology}
James~R. Munkres, \emph{Topology}, $2^{\textup{nd}}$ ed., Prentice Hall, 2000.

\bibitem{sion_1958}
Maurice Sion, \emph{On general minimax theorems}, Pacific Journal of
  Mathematics \textbf{8} (1958), no.~1, 171--176.

\bibitem{spinat_approach}
Xavier Spinat, \emph{A necessary and sufficient condition for approachability},
  Mathematics of Operations Research \textbf{27} (2002), no.~1, 31--44.

\bibitem{vieille_weak_approach}
Nicolas Vieille, \emph{Weak approachability}, Mathematics of Operations
  Research \textbf{17} (1992), no.~4, 781--791.

\end{thebibliography}
\end{document}